\pdfoutput=1
\RequirePackage{ifpdf}
\ifpdf 
\documentclass[pdftex]{sigma}
\else
\documentclass{sigma}
\fi

\numberwithin{equation}{section}

\newtheorem{Theorem}{Theorem}[section]
\newtheorem{Proposition}[Theorem]{Proposition}

\begin{document}

\newcommand{\arXivNumber}{1812.05535}

\renewcommand{\PaperNumber}{054}

\FirstPageHeading

\ShortArticleName{Interpolations between Jordanian Twists Induced by Coboundary Twists}

\ArticleName{Interpolations between Jordanian Twists\\ Induced by Coboundary Twists}

\Author{Andrzej BOROWIEC~$^{\dag^1}$, Daniel MELJANAC~$^{\dag^2}$, Stjepan MELJANAC~$^{\dag^3}$ and Anna PACHO{\L}~$^{\dag^4}$}
\AuthorNameForHeading{A.~Borowiec, D.~Meljanac, S.~Meljanac and A.~Pacho{\l}}

\Address{$^{\dag^1}$~Institute of Theoretical Physics, University of Wroclaw,\\
\hphantom{$^{\dag^1}$}~pl.~M.~Borna 9, 50-204 Wroclaw, Poland}
\EmailDD{\href{mailto:andrzej.borowiec@ift.uni.wroc.pl}{andrzej.borowiec@ift.uni.wroc.pl}}

\Address{$^{\dag^2}$~Division of Materials Physics, Ruder Bo\v{s}kovi\'c Institute,\\
\hphantom{$^{\dag^2}$}~Bijeni\v{c}ka~c.54, HR-10002~Zagreb, Croatia}
\EmailDD{\href{mailto:Daniel.Meljanac@irb.hr}{Daniel.Meljanac@irb.hr}}

\Address{$^{\dag^3}$~Division of Theoretical Physic, Ruder Bo\v{s}kovi\'c Institute,\\
\hphantom{$^{\dag^3}$}~Bijeni\v{c}ka~c.54, HR-10002~Zagreb, Croatia}
\EmailDD{\href{mailto:meljanac@irb.hr}{meljanac@irb.hr}}

\Address{$^{\dag^4}$~Queen Mary, University of London, Mile~End~Rd., London~E1~4NS, UK}
\EmailDD{\href{mailto:a.pachol@qmul.ac.uk}{a.pachol@qmul.ac.uk}}

\ArticleDates{Received February 15, 2019, in final form July 11, 2019; Published online July 21, 2019}

\Abstract{We propose a new generalisation of the Jordanian twist (building on the previous idea from [Meljanac S., Meljanac D., Pacho{\l} A., Pikuti\'c D., \textit{J.~Phys.~A: Math. Theor.} \textbf{50} (2017), 265201, 11~pages]). Obtained this way, the family of the Jordanian twists allows for interpolation between two simple Jordanian twists. This new version of the twist provides an example of a new type of star product and the realization for noncommutative coordinates. Real forms of new Jordanian deformations are also discussed. Exponential formulae, used to obtain coproducts and star products, are presented with details.}

\Keywords{twist deformation; Hopf algebras; coboundary twists; star-products; real forms}

\Classification{81T75; 16T05; 17B37; 81R60}

\section{Introduction}
\label{intro} Let $\mathcal{H}=( H,\Delta ,S,\epsilon) $ be a Hopf algebra and $F\in H\otimes H$ be a two-cocycle twist. Then new (twisted) Hopf algebra structure on the algebra $H$ with deformed coproduct and antipode is denoted by $\mathcal{H}^{F}=\big( H,\Delta ^{F},S^{F},\epsilon \big) $, where $\Delta ^{F}(\cdot )=F\Delta (\cdot)F^{-1}$.

For any invertible element $\omega \in H$ one can define new gauge equivalent two-cocycle twist $F_{\omega }=\big(\omega ^{-1}\otimes \omega^{-1}\big)F\Delta (\omega )$ which determines the third Hopf algebra $\mathcal{H}^{F_{\omega }}=\big(H,\Delta ^{F_{\omega }},S^{F_{\omega }},\epsilon \big) $. Notice that all three Hopf algebras share the same algebraic structure (multiplication). Its internal automorphism, defined by the similarity transformation: $\alpha (Z) =\omega Z\omega ^{-1}$, $Z\in H$ establishes, at the same
time, the isomorphism between two twisted Hopf algebras $\mathcal{H}^{F}\cong \mathcal{H}^{F_{\omega }}$ as illustrated on the following diagram
\begin{gather*}
\begin{array}{ccc}
H & \xrightarrow{\ \ \ \Delta ^{F}\ \ \ } & H\otimes H \\
& & \\
\alpha \downarrow & & \downarrow \alpha \otimes \alpha \\
& & \\
H & \xrightarrow{\ \Delta ^{F_\omega} } & H\otimes H,
\end{array}
\end{gather*}
i.e.,
\begin{gather*} 
\Delta ^{F}\circ \alpha = ( \alpha \otimes \alpha ) \circ \Delta ^{F_{\omega }}.
\end{gather*}

Consider now a (left) Hopf module algebra $\mathcal{A}=(A,\triangleright,\star )$ over the Hopf algebra $\mathcal{H}$ together with a (left) Hopf action $\triangleright \colon H\otimes A\rightarrow A$, where $\star $ denotes the multiplication in $A$. Changing multiplication $m_{\star }\equiv \star $ to $m_{\star _{F}}\equiv \star _{F}$:
\begin{gather}
a\star _{F}b=m_{\star _{F}}(a\otimes b) =m_{\star}\big[F^{-1}(\triangleright \otimes \triangleright )(a\otimes b)\big] \label{stp}
\end{gather}%
for $a,b\in A$, one gets new module algebra $\mathcal{A}_{F}=(A,\triangleright ,\star _{F})$ over $\mathcal{H}^{F}$ with the same action, i.e., the module structure remains the same. It is easy to see that any invertible element $\omega \in H$ provides an algebra isomorphism $\beta\colon (A,\star _{F_{\omega }})\rightarrow (A,\star _{F})$, where $\beta (a)\equiv \omega \triangleright a$, since $\omega \triangleright (a\star _{F_{\omega}}b)=(\omega \triangleright a)\star _{F}(\omega \triangleright b)$. It turns out that the invertible map $\beta $ intertwines between two modules in the following sense: $\alpha (Z)\triangleright \beta (a)=\beta (Z\triangleright a)$, i.e., the diagram
\begin{gather*}
\begin{array}{ccc}
H\otimes A & \xrightarrow{\ \ \ \triangleright\ \ \ } & A \\
& & \\
\alpha \otimes \beta \downarrow & & \downarrow \beta \\
& & \\
H\otimes A & \xrightarrow{\ \triangleright } & A \\
\end{array}
\end{gather*}
commutes. In particular, the coboundary twist $T_{\omega }=\big(\omega^{-1}\otimes \omega ^{-1}\big)\Delta (\omega )$ provides the Hopf algebra isomorphism $\mathcal{H}\cong \mathcal{H}^{T_{\omega }}$ as well as module algebra isomorphism $\mathcal{A}\cong \mathcal{A}_{T_{\omega }}$. For group-like $\omega $, it becomes an automorphism. It means that replacing a~twist by the gauge equivalent one leads to mathematically equivalent objects.

Let us explain this in the case of Jordanian deformations of Lie algebras.

Let's consider a Lie algebra $\mathfrak{g}$. Drinfeld's quantum groups are quantized universal enveloping algebras $U_{\mathfrak{g}}$~\cite{drinfeld} obtained by the methods of deformation quantization of Poisson Lie groups. More exactly, quantized objects are Hopf algebras corresponding to a given Lie bialgebra structure $(\mathfrak{g},r)$ \cite{etingof,etingof+}, which in turn can be determined by the classical $r$-matrix $r\in \mathfrak{g}\otimes \mathfrak{g}$ satisfying classical Yang--Baxter equation $[[r,r]]=0$ with $[[\,,\,]]$ being the so-called Schouten brackets. In the triangular case $r\in \mathfrak{g}\wedge \mathfrak{g}$ the quantization is provided by an~invertible, two-cocycle element $F_{r}\in U_{\mathfrak{g}}\otimes U_{\mathfrak{g}}[[\gamma ]]$ called a Drinfeld twist. Here $\gamma $ is a~formal parameter and $U_{\mathfrak{g}}[[\gamma ]]$ means topological completion in the topology of formal power series in $\gamma $ (see, e.g., \cite{chiari,Comtet, majid} for details). Thus the classical $r$-matrix can be recovered from the quantum $R$-matrix as follows
\begin{gather*}
R_{r}=F_{r}^{\tau}F_{r}^{-1}=1\otimes 1+\gamma r+o\big(\gamma ^{2}\big), 
\end{gather*}
where $\tau$ denotes the flip map: $F^\tau = F_{21}$. For example, in two dimensions there are only two (non-isomorphic) Lie algebra structures: Abelian $\mathfrak{ab}(2)=\{x,y\colon [x,y] =0\}$ and non-Abelian $\mathfrak{an}(2)=\{h,e\colon [h,e] =e\}$.\footnote{It corresponds to the Lie group ``$ax+b$'' of affine transformations of the real (complex) line.} The corresponding Lie bialgebra structures are given by the Abelian: $r_{Ab}=x \wedge y$ or Jordanian $r_{J}=h\wedge e$ classical $r$-matrices. Embedding one of these algebras in some higher-dimensional Lie algebra $\mathfrak{g}$ as Lie subalgebra provides the twist quantization of $U_{\mathfrak{g}}$: Abelian or Jordanian. In the latter case, deformation can be realized by the Jordanian twist of the form
\begin{gather}
F_{J}=\exp ( \ln ( 1+\gamma e ) \otimes h ). \label{FJ}
\end{gather}
This form of the twist first appeared in \cite{Og} and then a symmetrised form (i.e., where $r_{J}$ appears in the first order in expansion of the twist) was proposed in \cite{tolstoy2, tolstoy1}\footnote{These techniques have been used before in the supersymmetric case in order to unitarize super Jordanian twist, see \cite{1, 2, 3}.}
\begin{gather} \label{F_Tol}
F_{T}=\exp \left( \frac{\gamma }{2}(he\otimes 1+1\otimes he)\right) \exp\left( \ln (1+\gamma e)\otimes h\right) \exp \left( -\frac{\gamma }{2}\Delta(he)\right).
\end{gather}
It was obtained by applying the coboundary twist to~\eqref{FJ}. There are many possible $r$-symmet\-ri\-sa\-tions for twists. This is due to the fact that twist deformation is defined up to the so-called gauge transformation (in Drinfeld's terminology), i.e., many twists can provide isomorphically equivalent Hopf algebraic deformations, if they differ by the coboundary twist (see, e.g.,~\cite{majid}).

Jordanian deformations have been of interest for quite some time. For example, Jordanian deformations of the conformal algebra were considered in \cite{PRD-conformal}, as well as in \cite{Ballesteros-Jordconf+,Ballesteros-Jordconf++,Ballesteros-Jordconf}. In \cite{Ballesteros-Jordconf+,Ballesteros-Jordconf++,Ballesteros-Jordconf} the deformations of Anti de Sitter and de Sitter algebras were also investigated. Jordanian deformations have been considered in applications in ${\rm AdS}/{\rm CFT}$ correspondence \cite{MatYosh, MatYosh++,MatYosh+,MatYosh+++}, as integrable deformations of sigma models in relation to deformation of ${\rm AdS}_5$ and supergravity~\cite{tongeren+,tongeren}. Jordanian twists have been applied in deformation of spacetime metrics~\cite{ijgmmp2016}, Maxwell equations and dispersion relations~\cite{jhep17}, as well as classical and quantum mechanics \cite{1903_zagreb}. Gauge theories under Jordanian deformation were also investigated in~\cite{DJP_SIGMA2012}.

The question we would like to address in the present paper is if the mathematical equivalences of coboundary twists, in some physically relevant situations, give rise, to some extent, to physically inequivalent descriptions. For example, when one considers the star product quantization, realizations for noncommutative coordinates, or the form of Heinsenberg algebra.

To this aim we embed our two-dimensional Lie algebra $\mathfrak{an}(2)$ into some bigger Lie algebra which has some potential application in physics. In the present paper we focus our attention on the Lie algebra $\mathfrak{g}=\{P_{\mu },D\}$, generated by momenta $P_{\mu }$ (spacetime translations in $n$-dimensions where $\mu =0,1,\dots,n-1$) and dilatation generators $D$ with the following commutation relations
\begin{gather*}
[ P_{\mu },D] =P_{\mu },\qquad [ P_{\mu },P_{\nu } ] =0.
\end{gather*}
This algebra can be considered as a subalgebra of some bigger Lie algebra, e.g., Poincar\'e--Weyl, de Sitter, etc. This embedding is realized by choosing two elements $\{h,e\}$ as $h=-D$ and $e=P$ ($P$ can be taken as any of $P_{\mu }$ and the formulae in the next Section \ref{II} will hold). However, for convenience, we choose the following notation: $P=v^{\alpha }P_{\alpha }$ where $v^{\mu }$ is the vector on Minkowski spacetime $\mathcal{M}_{1,n-1}$ in $n$-dimensions such that $v^{2}=v^{\alpha }v_{\alpha }\in \{{-1,0,1\}}$. For the correspondence with the $\kappa $-Minkowski spacetime \cite{LukRuegg, LukTol} we choose the deformation parameter as $\gamma =- \frac{1}{\kappa }$.

Our main aim in this work is to analyse quantum deformations corresponding to gauge equivalence of Jordanian twists (\ref{FJ}), extending \cite{MMPP} and applying \cite{mercati}.

This paper is a sequel to~\cite{MMPP}, where we introduced a generalised form of $r$-symmetrised twist interpolating between Jordanian twists. The main formulae are recalled here as part of the next Section~\ref{FL}. In Section~\ref{FR} we propose another generalised form of $r$-symmetrised Jordanian twist ($F_{R,u}$) providing the interpolation between Jordanian twists as well. We present the corresponding Hopf algebra deformation, the star product form and the realization of the noncommutative coordinates. Section~\ref{relations} presents a relation between the two generalisations of $r$-symmetric twists $F_{L,u}$ and $F_{R,u}$, including the relation between the corresponding quantum $R$-matrices. In Section~\ref{star} the $*$-structure and unitarity of the twists is analysed. We finish with brief conclusions which are followed by a series of appendices complementing the results presented in the main part of the paper. They are devoted to the explanation of the exponential formulae which are obtained from twist realization of deformed coordinate functions. Some applications for calculations of wave packets star products as well as coproducts of momenta are also considered.

\section{Two families of twists interpolating between Jordanian twists}\label{II}

In our previous work \cite{MMPP} we have proposed a simple generalisation of the locally $r$-symmetric Jordanian twist \eqref{F_Tol}, resulting in the one-parameter family interpolating between Jordanian twists. All the proposed twists differed by the coboundary twists and produced the same Jordanian deformation of the corresponding Lie algebra. We have proposed a~way, by introducing an additional parameter~$u$, of interpolating between the two Jordanian twists
\begin{gather} \label{F0}
F_{0}=\exp \left( -\ln \left( 1-\frac{1}{\kappa }P\right) \otimes D\right)
\end{gather}
with the logarithm on the left side of the tensor product (cf.\ with~\eqref{FJ}) and
\begin{gather*} 
F_{1}=(F_{0})^\tau|_{-\frac{1}{\kappa }}=\exp \left( -D\otimes \ln \left( 1+\frac{1}{\kappa }P\right) \right)
\end{gather*}
with the logarithm on the right tensor factor and also with the changed sign of the deformation parameter $\kappa $, recall $\tau $ denotes the flip map: $F^\tau = F_{21}$. One can symmetrise these simple Jordanian twists, into a so-called $r$-symmetric form, such that at the first order in the expansion of the twist one gets the classical $r$-matrix corresponding to the given deformation. For the Jordanian deformations it is always $r_{J}$. In this paper we want to present another type of such interpolation.

First, in the below section, we recall few main formulae from \cite{MMPP} and then, in Section~\ref{FR}, we shall propose another interpolation.

\subsection[$F_{L,u}$ family of twists with dilatation on the left]{$\boldsymbol{F_{L,u}}$ family of twists with dilatation on the left}\label{FL}

The $r$-symmetric version of the Jordanian twist \eqref{F0} was introduced in~\cite{tolstoy1} and it was obtained from the coboundary twist $T_\omega$ by choosing $\omega _{0}=\exp \big({-} \frac{1}{2\kappa }DP\big)$. The formula for $F_{T}$ follows directly from: $\big(\omega _{0}^{-1}\otimes \omega_{0}^{-1}\big)F_{0}\Delta (\omega _{0})$.

In \cite{MMPP} we have introduced its generalisation in the form of one parameter family interpolating between Jordanian twists $\forall\, u\in \mathbb{R}$
\begin{gather}
F_{L,u}=\exp \left( \frac{u}{\kappa }(DP\otimes 1+1\otimes DP)\right) \exp\left( -\ln \left(1-\frac{1}{\kappa }P\right)\otimes D\right)\nonumber\\
\hphantom{F_{L,u}=}{} \times\exp\left( \Delta\left(-\frac{u}{\kappa }DP\right)\right),\qquad u\in \mathbb{R}. \label{FLu}
\end{gather}
This generalisation simply corresponds to a modification of $\omega _{0}$ to $\omega _{L}=\exp \big( {-}\frac{u}{\kappa }DP\big) $ and then $F_{L,u}=\big(\omega _{L}^{-1}\otimes \omega _{L}^{-1}\big)F_{0}\Delta (\omega _{L})$ still differs from $F_{0}$ only by the coboundary twist $T_{\omega _{L}}$. For this reason the cocycle condition (see, e.g.,~\cite{majid}) for $F_{L,u}$ is automatically satisfied. For $u=\frac{1}{2}$~\eqref{F_Tol} is recovered.

Note that now $F_{L,u}$ contains two parameters: one real parameter $u$ and the other $\kappa $ -the formal deformation parameter.

The reduction of $F_{L,u}$, for certain values of the parameter $u$, to $F_{0}$ (for $u=0$) and to $F_{1}$ (for $u=1$) were discussed in~\cite{MMPP}.

\subsubsection{Hopf algebra}\label{HAFL}

The deformation of the Hopf algebra $U_{\mathfrak{g}} ( \mu ,\Delta ,\epsilon ,S ) $ of the universal enveloping algebra of $\mathfrak{g}$ given by the twist element $F\in U_{\mathfrak{g}}\otimes U_{\mathfrak{g}}[[\frac{1}{\kappa }]]$ into $U_{\mathfrak{g}}^{F}\big( \mu ,\Delta^{F},\epsilon ,S^{F}\big) $ is provided by the deformation of the
coproduct and antipode maps as follows: $\Delta^F (Z) =F\Delta(Z) F^{-1}$, $S^{F}(Z)=\big[\mu ( (1\otimes S)F ) ]S(Z)[\mu\big( (S\otimes 1)F^{-1}\big) \big]$, where $Z\in
\mathfrak{g}$. The coproducts, star products and realizations depend explicitly on the parameter $u$ as well as on the parameter of deformation~$\kappa$.

We recall the coalgebra sector of the Hopf algebra $U_{\mathfrak{g}}^{F_{0}}$ for the deformation with the twisting element $F_{0}$ from~\cite{BP}
\begin{gather*}
\Delta ^{F_{0}}(P_{\mu }) =P_{\mu }\otimes 1+\left(1-\frac{1}{\kappa }P\right)\otimes P_{\mu } ,\qquad \Delta ^{F_{0}}(D)=D\otimes 1+\frac{1}{1-\frac{1}{\kappa }P}\otimes D, \\
S^{F_{0}}(P_{\mu }) =\frac{-P_{\mu }}{1-\frac{1}{\kappa }P},\qquad S^{F_{0}}(D)=-\left(1-\frac{1}{\kappa }P\right)D .
\end{gather*}
The change of the twist by the coboundary twist $T_{\omega }$ provides a new presentation for the Hopf algebra, and can be transformed by $\big( \alpha ^{-1}\otimes \alpha ^{-1}\big) \circ \Delta ^{F_{0}}\circ \alpha =\Delta ^{F_{\omega }}$ where $\alpha (Z) =\omega Z\omega ^{-1}$, $Z\in H $ and $\alpha ^{-1}(Z) =\omega ^{-1}Z\omega $ for chosen $\omega. $ For $F_{L,u}$ one needs to take $\omega _{L}$.

Alternatively, the deformed coproducts and antipodes can be calculated directly from the definition $\Delta Z =F\Delta(Z) F^{-1}$. The coalgebra sector of the Hopf algebra $U_{\mathfrak{g}}^{F_{L,u}}$ for the deformation with $F_{L,u}$, recalled from~\cite{MMPP}, is as follows
\begin{gather*}
\Delta ^{F_{L,u}}(P_{\mu }) =\frac{P_{\mu }\otimes \big(1+\frac{u}{\kappa }P\big)+\big(1-\frac{(1-u)}{\kappa }P\big)\otimes P_{\mu }}{1\otimes 1+u(1-u)\left( \frac{1}{\kappa }\right) ^{2}P\otimes P} ,\\ 
\Delta ^{F_{L,u}}(D) =\left( D\otimes \frac{1}{1+u\frac{1}{\kappa }P}+\frac{1}{1-(1-u)\frac{1}{\kappa }P}\otimes D\right) \left( 1\otimes 1+u(1-u)\left(
\frac{1}{\kappa }\right) ^{2}P\otimes P\right) .
\end{gather*}
The coproduct is coassociative. The antipodes are given by
\begin{gather*}
S^{F_{L,u}}(P_{\mu }) =-\frac{P_{\mu }}{1-(1-2u)\frac{1}{\kappa }P},\\ 
 S^{F_{L,u}}(D) = -\left(\frac{1-(1-2u)\frac{P}{\kappa}}{1+\frac{u}{\kappa}P}\right) D\left(1+\frac{u}{\kappa}P\right).
\end{gather*}

\subsubsection{Coordinate realizations and star product}\label{CoordFL}

As a Hopf module algebra for $U_{\mathfrak{g}}$ we choose the algebra of smooth (complex valued) functions on a space time (i.e., $A=C^{\infty }(\mathbb{R}^{n})\otimes \mathbb{C}$ with an obvious algebraic structures determined by point-wise multiplication and addition). This algebra includes spacetime coordinates~$x^{\mu }$ (where~$x^{\mu }$ are considered as generators of $n$-dimensional Abelian Lie algebra). A natural action of~$\mathfrak{g}$ on~$A$ (i.e., action of the momenta and the dilatation operators on coordinates) is defined, in terms of Weyl--Heisenberg algebra generators $[\partial_\mu, x^\nu]=\delta^\nu_\mu$,\footnote{This algebra has structure of a smash product and can be arranged in a Hopf algebroid structure \cite{DSRsmash,bp_jpa2016,rina-PLA2013,IJMPA2014}.} by
\begin{gather}
P_{\mu }\triangleright f(x)=-{\rm i}\frac{\partial }{\partial x^{\mu }}\triangleright f(x)=-{\rm i}\frac{\partial f\left( x\right) }{\partial x^{\mu }}
\qquad \text{and} \qquad D\triangleright f(x)=x^{\mu }\frac{\partial f(x)}{\partial
x^{\mu }}. \label{action}
\end{gather}
The algebra of coordinates becomes noncommutative due to the twist deformation once the usual multiplication is replaced by the star product multiplication (star product quantization)~\eqref{stp} for any $f,g\in C^{\infty }(\mathbb{R}^{n})$. This star product is associative (due to the fact that the twist~$F_{L,u}$ satisfies the cocycle
condition).

When we choose the functions to be exponential functions ${\rm e}^{{\rm i}k\cdot x}$ and ${\rm e}^{{\rm i}q\cdot x}$, then we can define new function $\mathcal{D}_{\mu }(u;k,q)$
\cite{hreal, mmss1,mmss2,mmss3,svrtan}:
\begin{gather}
{\rm e}^{{\rm i}k\cdot x}\star {\rm e}^{{\rm i}q\cdot x}=m\big[ \mathcal{F}^{-1}(\triangleright
\otimes \triangleright )\big({\rm e}^{{\rm i}k\cdot x}\otimes {\rm e}^{{\rm i}q\cdot x}\big)\big] ={\rm e}^{{\rm i}\mathcal{D}_{\mu }(u;k,q)x^{\mu }}, \label{stpFL}
\end{gather}%
where $k,q\in \mathcal{M}_{1,n-1}$ (in $n$-dimensional Minkowski spacetime). One can calculate explicitly, see Appendices~\ref{App1} and~\ref{App3}, that in the case of twist $F_{L,u}$ the function $\mathcal{D}_{\mu }(u;k,q)$ is given by
\begin{gather}
\mathcal{D}_{\mu }(u;k,q)=\frac{k_{\mu }\big(1+\frac{u}{\kappa }(v\cdot q)\big)+\big(1-\frac{(1-u)}{\kappa }(v\cdot k)\big)q_{\mu }}{1+\frac{u(1-u)}{\kappa ^{2}}(v\cdot k)(v\cdot q)}. \label{DFL}
\end{gather}%
Directly from the twist we can also calculate the coordinate realizations
\cite{DSRsmash,govindarajan,kov2,EPJC2015,rina-PLA2013,IJMPA2014,pikuticEPJC2017} and for the $F_{L,u}$ twist they have the following form\footnote{For simplicity the matrix notation ${\phi}_{\alpha}{}^{\mu }(P)$ is written as $\phi_{\alpha}^{ \mu } (P)$ throughout the paper.}
\begin{gather}
\hat{x}^{\mu }=m\big[ F_{L,u}^{-1}(\triangleright \otimes 1)(x_{\mu
}\otimes 1)\big] =x^{\mu }\left(1+\frac{u}{\kappa }P\right)+\frac{{\rm i}}{\kappa }v^{\mu}(1-u)D\left(1+\frac{u}{\kappa }P\right) \nonumber\\
\hphantom{\hat{x}^{\mu }}{} =\left(x^{\mu }+\frac{{\rm i}}{\kappa }v^{\mu }(1-u)D\right)\left(1+\frac{u}{\kappa }P\right)
 =x^{\alpha } \phi_{\alpha}^{ \mu } (P),\label{x-from-twist}
\end{gather}
where $v^{\mu }$ is the vector on Minkowski spacetime $\mathcal{M}_{1,n-1}$ in $n$-dimensions such that $v^{2}\in \{-1,0,1\}$ as before. These realizations are also discussed in \cite{kovmel,kmps, mkj,stojic1,mssg,stojic2}.

\subsection[$F_{R,u}$ family of twists with dilatation on the right]{$\boldsymbol{F_{R,u}}$ family of twists with dilatation on the right}\label{FR}

In this paper we want to introduce another version of the generalised Jordanian twist
\begin{gather}
F_{R,u}=\exp \left(\frac{u}{\kappa }(PD\otimes 1+1\otimes PD)\right)\exp \left(\!-\ln \left(1-\frac{1}{\kappa }P\right)\otimes D\right)\exp \left(\Delta\left(-\frac{u}{\kappa }PD\right)\right),\!\!\!\!\! \label{FRu}
\end{gather}
where $u$ is a real parameter $u\in \mathbb{R}.$ We point out that the
sub-index $R$ refers to the position of the dilatation generator, it is on
the right with respect to momenta generators $P$. Due to the position of the
dilatation operator with respect to momenta, this introduces a different
formula than the one considered in \cite{MMPP} and which was recalled in the previous
Section~\ref{FL}, i.e., $F_{L,u}$~\eqref{FLu}.
For the parameter $u=\frac{1}{2}$ the twist $F_{R,\frac{1}{2}}$ is
$r$-symmetric, but is not equal to $F_{L,\frac{1}{2}}$. These two twists differ at subleading order, see~\eqref{rel} in the following section.
The form of the family of twists $F_{R,u}$ can also be easily obtained from
the simple Jordanian twist $F_{0}$ by the transformation with the coboundary
twist $T_{\omega _{R}}$ however this time with the element $\omega _{R}=\exp \big( {-}\frac{u}{\kappa }PD\big) $. The twist $F_{R,u}$, $\forall\, u$, satisfies the normalization and cocycle conditions.

For $u=0$, twist $F_{R,u}$ simplifies to $F_{0}$, easily seen by just
plugging in $u=0$ in the equation \eqref{FRu}, and for $u=1$, it simplifies to the
twist $F_{1}$. Hence $F_{R,u}$ provides another way of interpolating between $F_{0}$ and $%
F_{1}.$

\subsubsection{Hopf algebra}\label{HAFR}

The coalgebra sector of the Hopf algebra $U_{\mathfrak{g}}^{F_{R,u}}$ for
the deformation with $F_{R,u}$ can also be calculated and has the form
\begin{gather}
\Delta ^{F_{R,u}}(P_{\mu }) =\frac{P_{\mu }\otimes \big(1+u\frac{1}{\kappa }
P\big)+\big(1-(1-u)\frac{1}{\kappa }P\big)\otimes P_{\mu }}{1\otimes 1+u(1-u)\left(\frac{1}{\kappa }\right) ^{2}P\otimes P}, \label{copPFR} \\
\Delta ^{F_{R,u}}(D) =\left( 1\otimes 1+\left( \frac{1}{\kappa }\right)
^{2}u(1-u) P\otimes P\right) \nonumber\\
\hphantom{\Delta ^{F_{R,u}}(D) =}{}\times \left( D\otimes \frac{1}{1+u\frac{1}{\kappa }P}+\frac{1}{1-(1-u) \frac{1}{\kappa }P}\otimes D\right).\label{copDFR}
\end{gather}
Antipodes are
\begin{gather}
S^{F_{R,u}}(P_{\mu }) =\frac{-P_{\mu }}{1-(1-2u)\frac{1}{\kappa }P}, \\
S^{F_{R,u}}(D) =-\left(1-(1-u)\frac{P}{\kappa}\right)D \left( \frac{1-(1-2u)\frac{P}{\kappa}}{1-(1-u)\frac{P}{\kappa}}\right).\label{SDFR}
\end{gather}

\subsubsection{Coordinate realizations and star product}

The inverse of the family of twists $F_{R,u}^{-1}$ provides another (new) star
product between the functions \eqref{stp}. If we choose our functions to be
exponential functions ${\rm e}^{{\rm i}k\cdot x}$ and ${\rm e}^{{\rm i}q\cdot x}$, then we can
define new functions $\mathcal{D}_{\mu }(u;k,q)$ and $\mathcal{G}(u;k,q)$ in
the following way \cite{hreal, mercati}
\begin{gather}
{\rm e}^{{\rm i}k\cdot x}\star {\rm e}^{{\rm i}q\cdot x} = m\big[F_{R,u}^{-1}(\triangleright \otimes
\triangleright )\big({\rm e}^{{\rm i}k\cdot x}\otimes {\rm e}^{{\rm i}q\cdot x}\big)\big]={\rm e}^{{\rm i}\mathcal{D}_{\mu
}(u;k,q)x^{\mu }+{\rm i}\mathcal{G}(u;k,q)} \notag\\
\hphantom{{\rm e}^{{\rm i}k\cdot x}\star {\rm e}^{{\rm i}q\cdot x}}{} = {\rm e}^{{\rm i}\mathcal{D}_{\mu }(u;k,q)x^{\mu }}\frac{1}{1+\frac{u(1-u)}{\kappa ^{2}}(v\cdot k) (v\cdot q) },\label{stpFR}
\end{gather}
where $k,q\in \mathcal{M}_{1,n-1}$. Note the difference in the terms on the right hand side between the formula above and the one for $F_{L,u}$ in \eqref{stpFL}. One can calculate, see Appendices~\ref{App2} and~\ref{App3}, that in the case of the twist $F_{R,u}^{-1}$ the function $\mathcal{D}_{\mu }(u;k,q)$ is the same as in equation~\eqref{DFL}. Note that the function $\mathcal{D}_{\mu }(u;k,q)$ can be seen as rewriting the coproduct $\Delta (P_{\mu })$ without using the tensor product notation (denoting left and right leg by $k$ and $q$ respectively). Therefore the relation between the coproduct $\Delta (P_{\mu })$ and the function $\mathcal{D}_{\mu }(u;k,q)$
is given by
\begin{gather*}
\Delta (P_{\mu })=\mathcal{D}_{\mu }(u;P\otimes 1,1\otimes P),\label{Delta-p-D}
\end{gather*}
hence $\Delta (P_{\mu })$ uniquely determines $\mathcal{D}_{\mu }(u;k,q)$, as in the case of $F_{L,u}$ twist.

The additional function on the right hand side of \eqref{stpFR} has the following explicit form
\begin{gather} \label{GFR}
\mathcal{G}(u;k,q)=i\ln \left(1+\frac{u(1-u)}{\kappa ^{2}}( v\cdot k) (v\cdot q) \right).
\end{gather}
We refer the reader to the Appendices~\ref{App2},~\ref{App3} and~\cite{mercati} for further the details of these calculations. Note that in the case of the generalisation of the twist~$F_T$~\eqref{F_Tol} into $F_{L,u}$~\eqref{FLu}, the function $\mathcal{G}(u;k,q)=0$ (see also~\cite{MMPP}).

Realizations of noncommutative coordinates $\hat{x}^{\mu }$ can be generally expressed in terms of Weyl--Heisenberg algebra generated by $x^{\mu }$ and $P_{\mu}$ (commutative variables).

The realization obtained from the $F_{R,u}$ twist has the new general form
\begin{gather*}
\hat{x}^{\mu }=x^{\alpha }\phi_\alpha^\mu(P)+\chi^\mu (P).
\end{gather*}
Note that the part $\chi^\mu (P) $ was not present in the case of $F_{L,u}$ twist, see equation \eqref{x-from-twist} in Section~\ref{CoordFL}, and also~\cite{MMPP}.

Noncommutative coordinates $\hat{x}^{\mu }$, corresponding to the twist $F_{R,u}$, are given by
\begin{gather}
\hat{x}^{\mu } =m\big[F_{R,u}^{-1}(\triangleright \otimes 1)(x^{\mu }\otimes 1)\big] =x^{\mu }\left(1+\frac{u}{\kappa }P\right)+\frac{{\rm i}}{\kappa }v^{\mu }(1-u)\left(1+\frac{u}{\kappa }P\right)D \nonumber\\
\hphantom{\hat{x}^{\mu }}{} =\left(x^{\mu }+\frac{{\rm i}}{\kappa }v^{\mu } ( 1-u ) D\right)\left(1+\frac{u}{\kappa}P\right)+u(1-u)\frac{{\rm i}}{\kappa ^{2}}v^{\mu}P.\label{x-from-FR}
\end{gather}

From the last line in the above formula one can read off explicitly the form of the functions $\phi_\alpha^\mu(P)$ and $\chi^\mu (P) $. In the case when $u=0$, we have $\hat{x}^{\mu }=x^{\mu }+{\rm i}\frac{1}{\kappa }v^{\mu }D$, while in the case when $u=1$, $\hat{x}^{\mu }=x^{\mu }\big(1+\frac{1}{\kappa }P\big)$.

The noncommutative coordinates $\hat{x}^{\mu }$ satisfy 
\begin{gather}
\big[\hat{x}^{\mu },\hat{x}^{\nu }\big] = \frac{{\rm i}}{\kappa }\big(v^{\mu }\hat{x}^{\nu }-v^{\nu }\hat{x}^{\mu }\big), \nonumber\\ 
\big[P_{\mu },\hat{x}^{\nu }\big] = \left(-{\rm i}\delta _{\mu }^{\nu }+{\rm i}\frac{1}{\kappa }v^{\nu }(1-u)P_{\mu }\right)\left(1+u\frac{1}{\kappa }P\right). \label{heis}
\end{gather}

The above kappa deformed Weyl--Heisenberg algebra \eqref{heis} is obtained by using the realization \eqref{x-from-FR}. It turns out to be the same as in \cite{MMPP} where it was obtained from \eqref{x-from-twist}. Summa\-ri\-sing, the realizations~\eqref{x-from-FR} and~\eqref{x-from-twist} differ by the term $\chi^\mu(P)$ which does not change the form of~\eqref{heis}.

\section[Relations between two families $F_{L,u}$ and $F_{R,u}$]{Relations between two families $\boldsymbol{F_{L,u}}$ and $\boldsymbol{F_{R,u}}$}\label{relations}

The two twists $F_{L,u}$ \eqref{FLu} and $F_{R,u}$ \eqref{FRu} are obviously related by the coboundary twists $T_{\omega _{L}}$ and $T_{\omega _{R}}$ in the following sense
\begin{gather*}
F_{R,u}=\big(\omega _{R}^{-1}\omega _{L}\otimes \omega _{R}^{-1}\omega_{L}\big)F_{L,u}\Delta \big( \omega _{L}^{-1}\omega _{R}\big),
\end{gather*}
where $\omega _{L}=\exp \big( {-}\frac{u}{\kappa }DP\big) $ and $\omega _{R}=\exp \big( {-}\frac{u}{\kappa }PD\big)$. Hence
\begin{gather*}
F_{R,u}^{-1}=\exp \left(\Delta\left(\frac{u}{\kappa }PD\right)\right)\exp \left(-\Delta\left(\frac{u}{\kappa }DP\right)\right)F_{L,u}^{-1}\exp \left(\frac{u}{\kappa }(DP\otimes 1+1\otimes DP)\right)\\
\hphantom{F_{R,u}^{-1}=}{} \times \exp \left(-\frac{u}{\kappa }(PD\otimes 1+1\otimes PD)\right) 
\end{gather*}
and
\begin{gather}
F_{L,u}F_{R,u}^{-1}=\Delta ^{F_{L,u}}\left( \exp \left(u\frac{P}{\kappa }D\right)\exp
\left(-uD\frac{P}{\kappa }\right)\right) \exp \left(\frac{u}{\kappa }(DP\otimes 1+1\otimes
DP)\right)\nonumber\\
\hphantom{F_{L,u}F_{R,u}^{-1}=}{}\times \exp \left(-\frac{u}{\kappa }(PD\otimes 1+1\otimes PD)\right),\label{36}
\end{gather}
where we have used the homomorphism property of the coproduct and its deformed form.

We can use the equalities\footnote{Here $\binom{D}{n}$ denotes the generalised binomial coefficient: $\binom{D}{n}=\frac{D^{\underline{n}}}{n!}=\frac{D(D-1)(D-2)\ldots (D-(n-1))}{n(n-1)(n-2)\cdots 1}$.}
\begin{gather*}
{\rm e}^{u\frac{P}{\kappa}D}=\sum_{n=0}^\infty \frac{\left(u\frac{P}{\kappa}D\right)^n}{n!}=\sum_{n=0}^\infty \left(u\frac{P}{\kappa}\right)^n \binom{D}{n}
\end{gather*}
and
\begin{gather*}
{\rm e}^{uD\frac{P}{\kappa}}=\sum_{n=0}^\infty \frac{\left(uD\frac{P}{\kappa}\right)^n}{n!}=\sum_{n=0}^\infty \left(u\frac{P}{\kappa}\right)^n \binom{D-1}{n},
\end{gather*}
where we chose the order with $P$ generators on the left. Now taking the difference of these two expressions we obtain
\begin{gather}
{\rm e}^{u\frac{P}{\kappa}D}-{\rm e}^{uD\frac{P}{\kappa}}=\sum_{n=0}^\infty \left(u\frac{P}{\kappa}\right)^n \left[\binom{D}{n}-\binom{D-1}{n}\right]=\sum_{n=1}^\infty \left(u\frac{P}{\kappa}\right)^n \binom{D-1}{n-1}\nonumber\\
\hphantom{{\rm e}^{u\frac{P}{\kappa}D}-{\rm e}^{uD\frac{P}{\kappa}}}{}=u\frac{P}{\kappa}{\rm e}^{uD\frac{P}{\kappa}}.\label{39}
\end{gather}
Therefore, we find, after multiplying both sides of \eqref{39} by ${\rm e}^{-uD\frac{P}{\kappa}}$ from the right
\begin{gather}\label{epd}
{\rm e}^{u\frac{P}{\kappa}D}{\rm e}^{-uD\frac{P}{\kappa}}=1+u\frac{P}{\kappa}.
\end{gather}
Inserting \eqref{epd} into \eqref{36} we get
\begin{gather*}
F_{L,u}F_{R,u}^{-1}=\Delta^{F_{L,u}}\left(1+u\frac{P}{\kappa}\right)\left(\frac{1}{1\otimes 1+u\frac{P}{\kappa}\otimes 1}\right)\left(\frac{1}{1\otimes 1+u\cdot 1\otimes \frac{P}{\kappa}}\right)\\
\hphantom{F_{L,u}F_{R,u}^{-1}}{} =\frac{1}{1\otimes 1+u(1-u)\frac{P}{\kappa}\otimes \frac{P}{\kappa}},
\end{gather*}
which leads to
\begin{gather}\label{rel}
F_{R,u}^{-1}=F_{L,u}^{-1}\frac{1}{1\otimes 1+\frac{u(1-u) }{\kappa ^{2}}P\otimes P}.
\end{gather}
Note that the twists $F_{L,u}$ and $F_{R,u}$ agree in the leading order of the deformation parameter, but are different at higher orders. We point out that using star product in \eqref{stpFL} and star product in \eqref{stpFR} and
methods introduced in \cite{mercati} one can also obtain the above relation \eqref{rel}.

Also one can write an explicit formula for the relation between $R_{R,u}$ and $R_{L,u}$ quantum $R$-matrices. It has the following form
\begin{gather*}
R_{L,u}=\frac{1}{1\otimes 1+\frac{u(1-u) }{\kappa ^{2}}P\otimes P}R_{R,u}\left( 1\otimes 1+\frac{u(1-u) }{\kappa ^{2}}P\otimes P\right).
\end{gather*}

\subsubsection*{Hopf algebras}\label{relationHA}

The two twists $F_{L,u}$ and $F_{R,u}$ describe two presentations of Hopf algebra (with the same corresponding classical $r$-matrix). The coproducts and antipodes for momenta are the same
\begin{gather*}
\Delta ^{F_{L,u}}(P_{\mu })=\Delta ^{F_{R,u}}(P_{\mu }),\qquad S^{F_{L,u}}(P_{\mu})=S^{F_{R,u}}(P_{\mu }).
\end{gather*}
However this is not the case for dilatations
\begin{gather*}
\Delta ^{F_{L,u}}(D)\neq \Delta ^{F_{R,u}}(D),\qquad S^{F_{L,u}}(D)\neq S^{F_{R,u}}(D).
\end{gather*}

\subsubsection*{Coordinate realizations and star products}\label{relationCoord}

Comparing the realizations we also see the difference. Equation~\eqref{x-from-FR} has an extra term only dependent on momenta whereas~\eqref{x-from-twist} does not.

Similarly in the formulae for the star products, the one coming from $F_{R,u} $ has an addition in the form of the function $\mathcal{G}(u;k,q)$ which does not appear in the star product coming from~$F_{L,u}.$ Nevertheless, these two twists are only a~change by coboundary twists from $F_{0}$ and provide equivalent Hopf algebra deformations (but with different representations).

\section{Discussion on the real forms of the Jordanian deformations}\label{star}
For physical applications it is important to address the question if the symmetry Hopf algebras and the deformed Weyl--Heisenberg algebra can be endowed with (compatible) $*$-structures. In general, $*$-structure (real Lie algebra structure) can be introduced by an antilinear involutive anti-automorphism $*\colon \mathfrak{g}\rightarrow \mathfrak{g}$ acting on the complex Lie algebra $\mathfrak{g}$. Thus the real coboundary Lie bialgebra can be considered as a triple $(\mathfrak{g},*,r)$, where the skew-symmetric element $r$ is assumed to be anti-Hermitian, i.e.,\footnote{More detailed exposition of the background material for the present section, as well as more complete list of references, can be found in~\cite{blt_jhep2017}.}
\begin{gather*}
	r^{*\otimes *}=-r = r^\tau.
\end{gather*}

The $*$-operation extends, by the property $(XY)^*=Y^* X^*$ (i.e., as an antilinear antiautomorphism), to the enveloping algebra $U_\mathfrak{g}$, as well as to quantized enveloping algebra, making both of them associative $*$-algebras. Therefore, any quantized enveloping algebra admits a natural $*$-structure, inherited
from the corresponding Lie bialgebra structure, which can be preserved under twist deformation provided that the twisting element is $*$-unitary, i.e.,
\begin{gather}\label{utw}
F^{*\otimes *}=F^{-1}.
\end{gather}

More exactly, we recall that a complex Hopf algebra $\mathcal{H}= ( H,\Delta ,S,\epsilon ) $ endowed with an antilinear involutive anti-automorphism $*\colon H\rightarrow H$ is called a real Hopf algebra or Hopf $*$-algebra if the following compatibility conditions for coproducts, antipodes and counits are satisfied
\begin{gather*}
\Delta(X^*) = (\Delta(X))^{*\otimes *},\qquad S((S(X^*))^*) = X,\qquad \epsilon(X^*)=\overline{\epsilon(X)}\qquad\forall\, X\in H,
\end{gather*}
where the standard ('unflipped') way of $*$-operation acting on a tensor product (i.e., coproduct) is assumed
\begin{gather*}
(X\otimes Y)^* = X^* \otimes Y^*.
\end{gather*}

In general, the twist deformation of quasitriangular Hopf algebra $(\mathcal{H},R)$ give rise to quasitriangular Hopf algebra with new universal $R$-matrix $R^F=F^\tau R F^{-1}$. If the real form of quasitriangular Hopf algebra is twisted by unitary twist then the deformed Hopf algebra is also real quasitriangular Hopf algebra. More precisely,
if $(H, \Delta, S, \epsilon, R, *)$ is a quasitriangular Hopf $*$-algebra with $R$ being real universal $R$-matrix (i.e., satisfying $R^{*\otimes *}=R^\tau$) (resp.\ antireal if satisfying $R^{*\otimes *}=R^{-1}$ ), then for any unitary, normalized 2-cocycle twist $F=(F^{-1})^{\star\otimes\star}\in H\otimes H$ the quantized algebra $(H, \Delta^F, S^F, \epsilon, R^F, *)$ is a quasitriangular Hopf $*$-algebra such that $R^F=F^\tau R F^{-1}$ is real (resp.\ antireal).

In our case only two Jordanian twists $F_0$ and $F_1$ are unitary, i.e., satisfy the condition~\eqref{utw} provided that we assume the following $*$-conjugations
\begin{gather}\label{conjugation}
P_\mu^*=P_\mu,\qquad\mbox{and}\qquad D^*=-D,
\end{gather}
which are compatible with the commutation relations $[P_\mu, D]=P_\mu$, $[ P_\mu, P_\nu]=0$.\footnote{It can be observed that
these commutation relations determine the generator $D$ up an additive constant while each generator~$P_\mu$ can be multiplied by a constant. For this reason the conjugation $D^*=-D+c$, $P_\mu^*={\rm e}^{{\rm i}b}P_\mu$ ($c$, $b$ are real constants) would be admissible too. In fact, our Weyl--Heisenberg $*$-algebra realisation~\eqref{action} is compatible with the assumption that $D^*=-D-n$, $P_\mu^*=P_\mu$, where $n$ denotes the spacetime dimension.}
Therefore, the corresponding twisted deformations are the Hopf $*$-algebras with the same $*$-structure.
Instead, for generic value of the parameter $u\in\mathbb{R}$ two families of twists
are related to each other by
\begin{gather*}
F_{L,u}^{*\otimes *}=F_{R,u}^{-1} ,\qquad\big[\Delta^{F_{L,u}}(X)\big]^{*\otimes *}=\Delta^{F_{R,u}}(X^*) ,\qquad \big[S^{F_{L,u}}(X)\big]^*=S^{F_{R,1-u}}(X^*)|_{-\kappa}.
\end{gather*}
Therefore, the well-known purely Jordanian twists $F_0$ and $F_1$ remain only unitary with respect to the conjugation
(\ref{conjugation}).
This situation can change, if we use the method proposed by S.~Majid \cite[Proposition~2.3.7, p.~59]{majid} and admit deformation of the original $*$-structure. Before, one needs to check if the twist satisfies the following condition
\begin{gather}\label{Maj_cond}
(S\otimes S)\big(F^{*\otimes *}\big)= F^\tau ,
\end{gather}
where the antipode map $S$ and the conjugations $*$ are taken before deformation. If it does, then one can define new quantized $\dag$-structure:
\begin{gather*}
(\quad)^{\dag} = S^{-1}(U) (\quad)^* S^{-1}\big(U^{-1}\big),
\end{gather*}
which makes the twisted Hopf algebra real. Here $U$ is related with the twist by the formula $ U =\sum\limits_{i} f^{(1)}_{i} S \big( f^{(2)}_{i}\big)$ and
$ U^{-1} =\sum\limits_{i} S\big(f^{-(1)}_{i}\big) f^{-(2)}_{i}$ with the short cut notation for the twist introduced as follows: $F = \sum\limits_{i} f^{(1)}_{i} \otimes f^{(2)}_{i} $ and its inverse as: $F^{-1} = \sum\limits_{i} f^{-(1)}_{i} \otimes f^{-(2)}_{i} $.

\subsection[New quantized $\dag$-structure for $F_{L,u=\frac{1}{2}}$ twist]{New quantized $\boldsymbol{\dag}$-structure for $\boldsymbol{F_{L,u=\frac{1}{2}}}$ twist}

One can explicitly check that only for one value of the parameter $u$, namely $u=\frac{1}{2}$\footnote{Note that only for this value of the parameter $u$ the twists $F_{R,u}$ and $F_{L,u}$ are $r$-symmetric.} the condition~\eqref{Maj_cond} for the Majid's method is satisfied. By direct calculation one gets
\begin{gather*}
\text{l.h.s.}=(S\otimes S)\big(F_{L,\frac{1}{2}}^{*\otimes *}\big)=\big(F_{L,\frac{1}{2}}\big)\big|_{-\kappa}
\end{gather*}
and
\begin{align*}
\text{r.h.s.}&=(F_{L,\frac{1}{2}})^\tau=\exp \left( \frac{1}{2\kappa }(DP\otimes 1+1\otimes DP)\right) (F_0)^\tau \exp\left( \Delta\left(-\frac{1}{2\kappa }DP\right)\right)\\
& = \exp \left( \frac{1}{2\kappa }(DP\otimes 1+1\otimes DP)\right) (F_1)|_{-\kappa} \exp\left( \Delta\left(-\frac{1}{2\kappa }DP\right)\right)\\
& = \left[\exp \left( - \frac{1}{2\kappa }(DP\otimes 1+1\otimes DP)\right) F_1 \exp\left( \Delta\left(\frac{1}{2\kappa }DP\right)\right)\right]|_{-\kappa}\\
& = \left[\exp \left( \frac{1}{2\kappa }(DP\otimes 1+1\otimes DP)\right) F_0 \exp\left( - \Delta\left(\frac{1}{2\kappa }DP\right)\right)\right]|_{-\kappa}
= \big(F_{L,\frac{1}{2}}\big)|_{-\kappa}.
\end{align*}
Hence
\begin{gather*}
(S\otimes S)\big(F_{L,\frac{1}{2}}^{*\otimes *}\big)=\big(F_{L,\frac{1}{2}}\big)\big|_{-\kappa}= \big(F_{L,\frac{1}{2}}\big)^\tau.
\end{gather*}
Therefore
\begin{gather*}
X^\dag=-S^{F_{L,\frac{1}{2}}}(X^*), \qquad\forall\, X\in H.
\end{gather*}
Analogously, one can show that $F_{R,\frac{1}{2}}$ also satisfies the condition \eqref{Maj_cond}
\begin{gather*}
(S\otimes S)\big(F_{R,\frac{1}{2}}^{*\otimes *}\big)=\big(F_{R,\frac{1}{2}}\big)\big|_{-\kappa}= \big(F_{R,\frac{1}{2}}\big)^\tau
\end{gather*}
and the new star structure
\begin{gather*}
X^\dag=-S^{F_{R,\frac{1}{2}}}(X^*), \qquad\forall\, X\in H.
\end{gather*}
It shows that for $u={1\over 2}$ we can introduce an exotic Hopf $\dag$-algebra structure which, however, in the classical limit reduces to the standard $*$ one.

\subsection[Coboundary Hopf $*$-algebra]{Coboundary Hopf $\boldsymbol{*}$-algebra}

Actually, the best way to obtain coboundary unitary twist preserving a given $*$-Hopf algebra structure is by using a unitary coboundary element $\omega^*=\omega^{-1}$. As a particular example let us consider a one-parameter family of such elements
\begin{gather*}
\omega_{LR}=\exp\left(-\frac{u}{2\kappa}(DP+PD)\right).
\end{gather*}
Then the new generalised family of the Jordanian twists
\begin{gather*}
F_{LR,u}=\big(\omega_{LR}^{-1}\otimes \omega_{LR}^{-1} \big)F_{0}\Delta (\omega_{LR})
\end{gather*}
is unitary for any real parameter $u$ and satisfies
\begin{gather*}
F_{LR,u}^{*\otimes *}=F_{LR,u}^{-1}
\end{gather*}
with $P_\mu^*=P_\mu$, $D^*=-D $. Note that $F_{LR,u}$ is $r$-symmetric again only for $u=\frac{1}{2}$ and differs from both $F_{L,\frac{1}{2}}$ and $F_{R,\frac{1}{2}}$.

$F_{LR,u}$ provides, via the usual twist deformation, a new Hopf $*$-algebra
\begin{gather*}
\Delta ^{F_{LR,u}}(P_{\mu }) = \sqrt{ 1\otimes 1 + \frac{u(1-u)}{\kappa^2}P\otimes P} \Delta ^{F_{L,u}}(P_\mu) \frac{1}{\sqrt{ 1\otimes 1 + \frac{u(1-u)}{\kappa^2}P\otimes P}} \\
\hphantom{\Delta ^{F_{LR,u}}(P_{\mu })}{} =
\frac{P_{\mu }\otimes \big(1+\frac{u}{\kappa }P\big)+\big(1-\frac{(1-u)}{\kappa }P\big)\otimes P_{\mu }}{1\otimes 1+ \frac{u(1-u)}{\kappa^2 }P\otimes P} ,
\label{copFL}\\
\Delta ^{F_{LR,u}}(D) =\sqrt{ 1\otimes 1 + \frac{u(1-u)}{\kappa^2}P\otimes P} \Delta ^{F_{L,u}}(D) \frac{1}{\sqrt{ 1\otimes 1 + \frac{u(1-u)}{\kappa^2}P\otimes P}} \\
\hphantom{\Delta ^{F_{LR,u}}(D)}{} = \sqrt{ 1\otimes 1 + \frac{u(1-u)}{\kappa^2}P\otimes P}\left(D\otimes \frac{1}{1+\frac{u}{\kappa}P}+\frac{1}{1-\frac{(1-u)}{\kappa}P}\otimes D \right)\\
\hphantom{\Delta ^{F_{LR,u}}(D)=}{} \times\sqrt{ 1\otimes 1 + \frac{u(1-u)}{\kappa^2}P\otimes P}
\end{gather*}
and
\begin{gather*}
S^{F_{LR,u}}(P_{\mu }) =-\frac{P_{\mu }}{1-(1-2u)\frac{1}{\kappa }P}, \\ 
 S^{F_{LR,u}}(D) = -\sqrt{\frac{\big(1-(1-u)\frac{P}{\kappa}\big)}{\big(1+\frac{u}{\kappa}P\big)}\left(1-(1-2u)\frac{P}{\kappa} \right)} D \sqrt{\frac{\big(1+\frac{u}{\kappa}P\big)}{1-(1-u)\frac{P}{\kappa}}\left(1-(1-2u)\frac{P}{\kappa} \right)}. 
\end{gather*}
These formulae can be calculated by using the following relation between the twists $F_{LR,u}$, $F_{L,u}$ and $F_{R,u}$:
\begin{gather*}
F^{-1}_{LR,u} =F^{-1}_{L,u}\frac{1}{\sqrt{ 1\otimes 1 + \frac{u(1-u)}{\kappa^2}P\otimes P}} = F^{-1}_{R,u}\sqrt{ 1\otimes 1 + \frac{u(1-u)}{\kappa^2}P\otimes P}.
\end{gather*}
We can also provide the relation between quantum $R$-matrices
\begin{gather*}
R_{LR,u} =\sqrt{ 1\otimes 1 + \frac{u(1-u)}{\kappa^2}P\otimes P} R_{L,u} \frac{1}{\sqrt{ 1\otimes 1 + \frac{u(1-u)}{\kappa^2}P\otimes P}}.
\end{gather*}

\section{Conclusions}
Jordanian deformations have been of interest in some recent
literature \cite{jhep17,Ballesteros-Jordconf+,Ballesteros-Jordconf++,ijgmmp2016,DJP_SIGMA2012,Ballesteros-Jordconf,tongeren+, MatYosh,MatYosh++,MatYosh+,MatYosh+++,
1903_zagreb,PRD-conformal,pv, tongeren}. In our previous paper \cite{MMPP} we have studied the simple genera\-li\-sa\-tion of the locally $r$-symmetric Jordanian twist. Following that idea now we have found another possible way of interpolating between two Jordanian twists $F_{0}$ and $F_{1}$. In both cases, we have introduced one real valued parameter $u$ and obtained the family of Jordanian twists which provides interpolation between the original unitary Jordanian twist (for $u=0$) and its flipped version (for $u=1$, up to minus sign in the deformation parameter). If $D^* = - D$ and $P_\mu^*=P_\mu$ both twists~$F_0$ and~$F_1$ are $*$-unitary while those for $u\neq 0,1$ are not. Only in the case of $u={1\over 2}$ the interpolating twists provide the real Hopf algebra structures with the deformed $*$-structures (obtained by deformation techniques from~\cite{majid}). Also only for this value of the parameter~$u$ the twists~$F_{R,u}$ and $F_{L,u}$ and $F_{LR,u}$ are $r$-symmetric. However, using unitary coboundary elements $\omega^*=\omega^{-1}$ one can interpolate preserving the original real $*$-structure. The new family of Jordanian twists~$F_{R,u}$~\eqref{FRu}, as all Jordanian twists mentioned in this work, provides the so-called $\kappa $-Minkowski noncommutative spacetime and has the support in the Poincar\'e--Weyl or conformal algebras as deformed symmetries of this noncommutative spacetime.

In this paper, starting from the twist \eqref{FRu}, we have found deformed Hopf algebra symmetry \eqref{copPFR}--\eqref{SDFR}, star products \eqref{stpFR} and corresponding realizations for noncommutative coordinates~$\hat x^\mu$~\eqref{x-from-FR}. Noncommutative coordinates induce the deformation of Weyl--Heisenberg algebra~\eqref{heis} as well.
Even though both of the proposed twists, the previous one~$F_{L,u}$ from~\cite{MMPP} and the one introduced here~$F_{R,u}$ provide the $\kappa $-Minkowski spacetime and have the support in the Poincar\'{e}--Weyl or conformal algebras as deformed symmetries of this noncommutative spacetime, the realizations and star products they induce differ. The new type of star product~\eqref{stpFR} obtained here contains additional term depending only on momenta $\mathcal{G}(u;k,q)$. The additional term $\chi^\mu(P)$ also appears in the realization of noncommuting coordinates in~\eqref{x-from-FR}. However, the Weyl--Heisenberg algebra form is the same in both $F_{L,u}$ and $F_{R,u}$ deformed cases~\eqref{heis}.

Therefore, mathematically equivalent deformations, may lead to differences in the physical phenomena. Many authors \cite{jhep17,DJP_SIGMA2012,1903_zagreb} have used star products to discuss physical consequences of deformations. In those cases the two families of Jordanian twists $F_{L,u}$, $F_{R,u}$ would lead to different outcomes. It would be interesting, for example, to investigate the deformation of differential and integral calculus in the context of the new versions of Jordanian twists proposed here and their application to second quantization~\cite{Fiore}. Also the differences in realizations of noncommutative coordinates could lead to different physical predictions, see, e.g., \cite{arXiv:0912.3299} where the influence of different realizations on modified dispersion relations between energy and momenta as well as time delay parameter was investigated.

\appendix

\section[Exponential formula, normal ordering and Weyl--Heisenberg algebra]{Exponential formula, normal ordering\\ and Weyl--Heisenberg algebra}\label{App0}

We recall that in the simplest case the Weyl--Heisenberg algebra $W_1$ (over the field~$\mathbb{C}$ of complex numbers) can be defined abstractly as a universal associative and unital algebra over $\mathbb{C}$ with two generators $x$, $p$ satisfying the relation
\begin{gather*}
xp-px={\rm i},
\end{gather*}
where ${\rm i}\in\mathbb{C}$ stands for the imaginary unit. Usually, the generator $p$ can be identified with the derivative $-{\rm i}\,{\rm d}/{\rm d}x$. Such realization makes it easy to remember the canonical action of $W_1$ onto the space of polynomial functions in one variable $\mathbb{C}[x]$. Therefore, any element $a\in W_1$ admits a canonical presentation either in the form of differential operator $a=\sum\limits_{s=0}^{N_p}a_s(x)p^s$, where $a_s(x)\in \mathbb{C}[x]$ or in the normal ordered form $a=\sum\limits_{r=0}^{N_x}x^r a_r(p)$, where $a_r(p)\in \mathbb{C}[p]$, see, e.g.,~\cite{CR}. This algebra can be also defined as a smash product of two polynomial algebras, i.e., $W_1=\mathbb{C}[x]\rtimes \mathbb{C}[p]$, where~$\mathbb{C}[x]$ plays a role of module algebra over the Hopf algebra $\mathbb{C}[p]$.\footnote{For the definition see, e.g.,~\cite{bp_jpa2016} and references therein.} Using smash product definition one can naturally extend~$W_1$ to $\mathbb{C}[[x]]\rtimes\mathbb{C}[p]$, i.e., replacing the polynomial algebra $\mathbb{C}[x]$ by its $x$-adic extension $\mathbb{C}[[x]]$ of formal power series (see, e.g.,~\cite{Comtet}), which bears the structure of (left) $\mathbb{C}[[x]]$ module. Its topological completion provides the algebra~$\widehat W_1$ with elements having a~form $\sum\limits_{r=0}^{\infty}x^r a_r(p)$, $a_r(p)\in \mathbb{C}[p]$.

For our purposes, however, one needs further extension by introducing a new commuting formal variable $k$ and taking $\widehat W_1[[k]]$ with elements of the form
\begin{gather*} \sum_{r=0}^{\infty}x^r a_r(k,p)=\sum_{r, s=0}^{\infty}x^r k^s a_{r,s}(p) , \end{gather*}
where $ a_{r,s}(p)\in \mathbb{C}[p] $ as before and $ a_{r}(k,p)=\sum\limits_{s=0}^\infty a_{r,s}(p)k^s\in \mathbb{C}[p][[k]] $. Now we are in position to formulate the following

\begin{Proposition}\label{propSM1}
In this framework:
\begin{enumerate}\itemsep=0pt
\item[$i)$] For any element $\phi\equiv \phi(p)\in \mathbb{C}[p]$ there exists a unique element $\Phi\equiv\Phi(k,p)\in\mathbb{C}[p][[k]]$ such that
\begin{gather*} 
\exp({\rm i} k x \phi(p)) = \sum_{r=0}^\infty\frac{( {\rm i}x)^r}{r!} [\Phi(k, p)]^r
\end{gather*}
holds true in $\widehat W_1[[k]]$. One should notice that the last expression can be denoted as
\linebreak $ {:}\exp({\rm i} x \Phi(k, p)){:}$,
where ${:}\dots {:}$ denotes normal ordering of the generators $x$, $p$ $($i.e., $x$'s left from $p$'s$)$.
\item[$ii)$] Moreover
\begin{gather*}
\Phi(k, p)=J ( k,p )-p,
\end{gather*}
where
\begin{gather}\label{J2dim}
J(k,p)={\rm e}^{-{\rm i}kx\phi(p)}p{\rm e}^{{\rm i}kx\phi(p)}.
\end{gather}
\item[$iii)$] $J(k,p)$ turns out to be a unique $($formal$)$ solution of the $($formal$)$ partial differential equation
\begin{gather}\label{Jdif1}
\frac{\partial}{\partial k }J(k,p) =\phi ( J(k,p))
\end{gather}
with the boundary condition: $J(0,p) = p$.
\end{enumerate}
\end{Proposition}
\begin{proof} Since $\exp({\rm i} k x \phi(p))\in \widehat W_1[[k]]$, one can employ the adjoint action $\mathrm{ad}_p a=[p,a]$ to calculate
\begin{gather*}
\exp({\rm i} k x \phi(p)) = \sum_{n=0}^\infty \frac{({\rm i}x)^n}{n!}\Phi_n(k,p),\\
(\mathrm{ad}_p)^m \exp({\rm i} k x \phi(p)) = \sum_{n=m}^\infty \frac{({\rm i}x)^{n-m}}{(n-m)!}\Phi_n(k,p),\\
(\mathrm{ad}_p)^m \exp({\rm i} k x \phi(p))|_{x=0} = \Phi_m(k,p).
\end{gather*}
Hence
\begin{gather}\label{eq3}
\exp({\rm i} k x \phi(p))=\sum_{n=0}^\infty \frac{({\rm i}x)^n}{n!}\big((\mathrm{ad}_p)^n (\exp({\rm i} k x \phi(p))\big)\big|_{x=0}.
\end{gather}
On the other hand, introducing $J(k,p)$ by formula (\ref{J2dim}), one gets (cf.~(\ref{eq3}))
\begin{gather*}
\mathrm{ad}_p \big({\rm e}^{{\rm i}k x \phi(p)}\big) = \big[p,{\rm e}^{{\rm i}k x \phi(p)}\big]=
{\rm e}^{{\rm i}k x \phi(p)}\big({\rm e}^{-{\rm i}kx\phi(p)}p{\rm e}^{{\rm i}kx\phi(p)}-p\big) \\
\hphantom{\mathrm{ad}_p ({\rm e}^{{\rm i}k x \phi(p)})}{} = {\rm e}^{{\rm i}k x \phi(p)}(J(k,p)-p)= {\rm e}^{{\rm i}k x \phi(p)}\Phi(k,p),
\end{gather*}
i.e., $\Phi(k,p)=J(k,p)-p $. Similarly, by induction
\begin{gather*}
(\mathrm{ad}_p)^n \big({\rm e}^{{\rm i}k x \phi(p)}\big) = {\rm e}^{{\rm i}k x \phi(p)}(\Phi(k,p))^n.
\end{gather*}
It follows that
\begin{gather*}
\Phi_n(k,p)=(\mathrm{ad}_p)^n \big({\rm e}^{{\rm i}k x \phi(p)}\big)\big|_{x=0}=(\Phi(k,p))^n=(J(k,p)-p)^n,
\end{gather*}
hence
\begin{gather*}
\exp({\rm i} k x \phi(p))=\sum_{n=0}^\infty \frac{({\rm i}x)^n}{n!}(\Phi(k,p))^n= {:}\exp({\rm i} k x \Phi(k,p)){:} .
\end{gather*}
Finally, differentiation of (\ref{J2dim}) gives
 \begin{gather*}
\frac{\partial}{\partial k }J(k,p) ={\rm e}^{-{\rm i}kx\phi(p)}[-{\rm i}x,p]\phi(p){\rm e}^{{\rm i}kx\phi(p)}=
{\rm e}^{-{\rm i}kx\phi(p)}\phi(p){\rm e}^{{\rm i}kx\phi(p)}= \phi(J(k,p)).
\end{gather*}
This provides the equation~(\ref{Jdif1}) together with the boundary condition: $J(0,p) = p$.
\end{proof}

\section{The multidimensional case}\label{mult}

Replacing $W_1$ by the Weyl--Heisenberg algebra $W_N$ with $2N$ generators $x^\alpha$, $p_\beta$:
\begin{gather*}
x^\alpha x^\beta-x^\beta x^\alpha=p_\alpha p_\beta-p_\beta p_\alpha=0 , \qquad x^\alpha p_\beta-p_\beta x^\alpha={\rm i}\delta^\alpha_\beta,
\end{gather*}
where $\alpha, \beta=0,1,\ldots,N-1$, allows for the following generalisation~\cite{svrtan}

\begin{Proposition}$\!\!\!\!$\footnote{Proof of Proposition~\ref{propSM2} and it's generalisations will be given elsewhere.}\label{propSM2}
For arbitrary realization $\phi(p)$ for the noncommutative coordinates \linebreak $\hat{x}^{\mu }=x^{\alpha }\phi_{\alpha}^{ \mu }(p) $ it holds:
\begin{enumerate}\itemsep=0pt
\item[$i)$]
\begin{gather*}
\exp \big( {\rm i}k_{\alpha }\hat{x}^{\alpha }\big) ={:}\exp \big( {\rm i}x^{\alpha }\Phi_\alpha(k,p)\big){:},
\end{gather*}
where ${:}\dots{:}$ denotes normal ordering of the generators $x$, $p$ $($i.e., $x$'s left from $p$'s$)$.
\item[$ii)$]
\begin{gather*}
\Phi_\mu(k,p)={\rm e}^{-{\rm i}k_\alpha \hat{x}^\alpha}p_\mu {\rm e}^{{\rm i}k_\alpha \hat{x}^\alpha}-p_\mu= J_\mu(k,p) -p_\mu,
\end{gather*}
i.e.,
\begin{gather}\label{epe}
J_\mu(k,p)={\rm e}^{-{\rm i}k_\alpha \hat{x}^\alpha}p_\mu {\rm e}^{{\rm i}k_\alpha \hat{x}^\alpha}.
\end{gather}
\item[$iii)$] $J_{\mu }(k,p)$ satisfies
\begin{gather}\label{Jdif}
\frac{{\rm d}}{{\rm d}\lambda }J_{\mu } ( \lambda k,p ) =k_{\alpha}\phi^\alpha_{\mu} ( J (\lambda k,p ) )
\end{gather}
with the boundary condition: $J_\mu(0,p)=p_\mu$.
\end{enumerate}
\end{Proposition}

The last relation~\eqref{Jdif} can be shown in the following way. Applying $k_{\alpha }\frac{\partial }{\partial k_{\alpha }}$ on both sides to~\eqref{epe} we have
\begin{gather*}
{\rm e}^{-{\rm i}k\hat{x}}[-{\rm i}k\hat{x},p_{\mu }]{\rm e}^{{\rm i}k\hat{x}}=k_{\alpha }\frac{\partial }{\partial k_{\alpha }} ( J_{\mu }(k,p) ).
\end{gather*}

Using the form of the realization for $\hat{x}^{\mu }=x^{\alpha }\phi_{\alpha}^{\mu}(p)$, we obtain
\begin{gather*}
k_{\alpha }\frac{\partial }{\partial k_{\alpha }}( J_{\mu } (k,p)) =k_{\alpha }\phi^\alpha_{\mu }( J(k,p)).
\end{gather*}
In order to simplify this equation we introduce the change of variables: $k_{\alpha }\rightarrow \lambda k_{\alpha }$ and use the identity
\begin{gather*}
\lambda \frac{{\rm d}}{{\rm d}\lambda }=\lambda k_{\alpha }\frac{\partial }{\partial (\lambda k_{\alpha })}=k_{\alpha }\frac{\partial }{\partial k_{\alpha }},
\end{gather*}
which leads to \eqref{Jdif}.

Further generalisation admits more general class of realizations. If
\begin{gather*}
\hat{x}^{\mu }=x^{\alpha }\phi_\alpha^\mu(p) +\chi^\mu ( p)
\end{gather*}
for $\phi_\alpha^\mu(p)$ and $\chi^\mu (p)$ as in~\eqref{x-from-FR} then we have
\begin{gather*}
\exp \big( {\rm i}k_{\alpha}\hat{x}^{\alpha}\big) ={:}\exp \big( {\rm i}x^{\alpha} ( J_{\alpha }(k,p) -p_{\alpha } ) \big){:} \, \exp ( {\rm i}Q(k,p)).
\end{gather*}
Therefore we get
\begin{gather*}
{\rm e}^{-{\rm i}k_{\alpha }x^{\beta }\phi_{\beta}^{\alpha }(p) }p_{\mu }{\rm e}^{{\rm i}k\hat{x}}=J_{\mu }(k,p) {\rm e}^{{\rm i}Q(k,p) }.
\end{gather*}
Applying $k_{\alpha }\frac{\partial }{\partial k_{\alpha }}$ on both sides we finally obtain
\begin{gather*}
k_{\alpha }\frac{\partial }{\partial k_{\alpha }}Q(k,p) =k_{\alpha }\chi ^{\alpha }(J(k,p)),
\end{gather*}
i.e., after changing variables $k_{\alpha }\rightarrow \lambda k_{\alpha }$ we find
\begin{gather*}
\frac{{\rm d}}{{\rm d}\lambda }Q ( \lambda k,p ) =k_{\alpha }\chi ^{\alpha} ( J ( \lambda k,p ) ).
\end{gather*}
Note that if $\chi ^{\alpha }(p) \neq 0$ then
\begin{gather*}
{\rm e}^{{\rm i}k_{\alpha }\hat{x}^{\alpha}}\vartriangleright {\rm e}^{{\rm i}q_{\beta }x^{\beta}}={\rm e}^{{\rm i} J_{\alpha } ( k,q ) {x}^\alpha+{\rm i}Q ( k,q ) },
\end{gather*}
where the action is defined in \eqref{action2} below. For $k_{\alpha }=0$ it holds: $J_{\mu } ( 0,q ) =q_{\mu }$ and $Q ( 0,q ) =0.$ If $q_{\mu }=0$ then $J_{\mu }( k,0) =K_{\mu }( k) $ and $Q( k,0) =g( k)$.

\section{Star products}
\subsection[Star products for $F_{L,u}$: general formulas]{Star products for $\boldsymbol{F_{L,u}}$: general formulas}\label{App1}

If we start with realization
\begin{gather*}
\hat{x}^{\mu }=x^{\alpha }\phi_{\alpha}^{\mu }(P),
\end{gather*}
then, for action $\triangleright $, defined by
\begin{gather}\label{action2}
x^{\mu }\triangleright f(x)=x^{\mu }f(x),\qquad P_{\mu }\triangleright f(x)=-{\rm i}\frac{\partial f(x)}{\partial x^{\mu }},
\end{gather}%
it holds
\begin{gather*}
{\rm e}^{{\rm i}k\cdot \hat{x}}\triangleright 1 ={\rm e}^{{\rm i}k_{\mu }x^{\alpha }{\phi}_{\alpha}^{\mu }(P)}\triangleright 1={\rm e}^{{\rm i}K(k)\cdot x}, \\ 
{\rm e}^{{\rm i}k\cdot \hat{x}}\triangleright {\rm e}^{{\rm i}q\cdot x} ={\rm e}^{{\rm i}k_{\mu }x^{\alpha }{\phi}_{\alpha}^{\mu }(P)}\triangleright {\rm e}^{{\rm i}q\cdot x}={\rm e}^{{\rm i}J(k,q)\cdot x}, 
\end{gather*}
where functions $K_{\mu }(k)$ and $J_{\mu }(k,q)$ can be calculated from the following differential equations
\begin{gather} \label{diffK}
\frac{{\rm d} K^{\mu }( \lambda k)}{{\rm d}\lambda} ={\phi^{\mu}_{\alpha }}(K(\lambda k))k^\alpha, \\
\frac{{\rm d} J^{\mu }(\lambda k,q)}{{\rm d}\lambda} ={\phi^{\mu}_{\alpha }}(J(\lambda k,q)) k^\alpha, \label{diffJ}
\end{gather}
with boundary conditions $K_{\mu }(0)=0$ and $J_{\mu } ( k,0 ) =K_{\mu }( k) $, $J_{\mu }(0,q)=q^{\mu }$.

The star product is given by
\begin{gather*}
{\rm e}^{{\rm i}k\cdot x}\star {\rm e}^{{\rm i}q\cdot x}={\rm e}^{{\rm i}K^{-1}(k) \cdot \hat{x}
}\triangleright {\rm e}^{{\rm i}q\cdot x}={\rm e}^{{\rm i}J(K^{-1}(k) ,q)\cdot x}={\rm e}^{{\rm i}\mathcal{D}(k,q)\cdot x},
\end{gather*}
where
\begin{gather*} 
\mathcal{D}(k,q)=J\big(K^{-1}(k),q\big)
\end{gather*}
with the inverse function of $K_{\mu }(k) $ defined as $K_{\mu}\big( K^{-1}(k) \big) =K_{\mu }^{-1} ( K(k) ) =k_{\mu }.$

\subsection[Star products for $F_{R,u}$: general formulas]{Star products for $\boldsymbol{F_{R,u}}$: general formulas}\label{App2}
For a more general realization given by
\begin{gather*}
\hat{x}^{\mu }=x^{\alpha }\phi_{\alpha}^{\mu }(P)+\chi^{\mu }(P),
\end{gather*}
it holds
\begin{gather*}
{\rm e}^{{\rm i}k\cdot \hat{x}}\triangleright 1 ={\rm e}^{{\rm i}K(k)\cdot x+{\rm i}g(k)}, \\ 
{\rm e}^{{\rm i}k\cdot \hat{x}}\triangleright {\rm e}^{{\rm i}q\cdot x} ={\rm e}^{{\rm i}J(k,q)\cdot x+{\rm i}Q(k,q)}. 
\end{gather*}

$K(k)$, $J(k,q)$ satisfy the same differential equation as in Appendix~\ref{App1}. Similarly we can determine $g(k) $ and $Q(k,q)$ by differentiating with respect to $\lambda$ \cite{hreal, mercati}
\begin{gather} \label{diffg}
\frac{{\rm d}g(\lambda k)}{{\rm d}\lambda } = k\cdot \chi (K(\lambda k)), \\
\frac{{\rm d}Q(\lambda k,q)}{{\rm d}\lambda } = k\cdot \chi (J(\lambda k,q)).\label{diffQ}
\end{gather}
The boundary condition is $Q ( k,0 ) =g ( k )$, $g (0 ) =Q ( 0,q ) =0$. This gives \cite{hreal, mercati}
\begin{gather*} 
g(k)=Q(k,0) =\int_{0}^{1}{\rm d}\lambda [ k\cdot \chi (K(\lambda k)) ], \\
Q(k,q) =\int_{0}^{1}{\rm d}\lambda [ k\cdot \chi (J(\lambda k,q)) ]. 
\end{gather*}

The star product is
\begin{gather*} 
{\rm e}^{{\rm i}k\cdot x}\star {\rm e}^{{\rm i}q\cdot x} ={\rm e}^{{\rm i}K^{-1}(k)\cdot \hat{x}
-g(K^{-1}(k))}\triangleright {\rm e}^{{\rm i}q\cdot x} ={\rm e}^{{\rm i}J(K^{-1}(k),q)\cdot x+{\rm i}Q(K^{-1}(k),q)-{\rm i}Q(K^{-1}(k),0)}.
\end{gather*}
Now we take
\begin{gather}
\mathcal{D} ( k,q ) =J\big( K^{-1}(k) ,q\big) ,\nonumber \\
\mathcal{G}(k,q) =Q\big(K^{-1}(k),q\big)-Q\big(K^{-1}(k),0\big).\label{DJK}
\end{gather}
Therefore, it follows that
\begin{gather*}
{\rm e}^{{\rm i}k\cdot x}\star {\rm e}^{{\rm i}q\cdot x}={\rm e}^{{\rm i}\mathcal{D}(k,q)\cdot x+{\rm i}\mathcal{G}(k,q)}.
\end{gather*}

\subsection[$F_{L,u}$ and $F_{R,u}$: explicit calculations]{$\boldsymbol{F_{L,u}}$ and $\boldsymbol{F_{R,u}}$: explicit calculations}\label{App3}
Realizations of noncommutative coordinates for $F_{L,u}$ is \eqref{x-from-twist}
\begin{gather*}
\hat{x}^{\mu }=\left[ x^{\mu }+\frac{{\rm i}(1-u)}{\kappa }v^{\mu }D\right] \left(1+\frac{u}{\kappa }P\right),\qquad u\in [0,1],
\end{gather*}
and for $F_{R,u}$ \eqref{x-from-FR}
\begin{gather*}
\hat{x}^{\mu }=\left(x^{\mu }+\frac{{\rm i}(1-u) }{\kappa }v^{\mu }D\right)\left(1+\frac{u}{\kappa }P\right)+\frac{{\rm i}u(1-u)}{\kappa ^{2}}v^{\mu }P,\qquad u\in [0,1].
\end{gather*}
From these realizations the form of the function $\phi _{\alpha}^{\mu}(P)$ can be read as
\begin{gather*}
\phi_\alpha^{\mu }(P)=\left(\delta _{\alpha}^{\mu}-\frac{(1-u)}{\kappa }v^{\mu }P_{\alpha }\right)\left(1+\frac{u}{\kappa }P\right),
\end{gather*}
where $P=v^{\alpha }P_{\alpha }$ is used as a shortcut.

Note that it is the same for both realizations \eqref{x-from-twist}, \eqref{x-from-FR}. Therefore, both of these realizations have the same form of the functions $K(k)$, $K^{-1}(k)$, $J(k,q)$ and $\mathcal{D}(k,q)$, see below for the explicit calculations.

To obtain $K^{\mu }(k)$ and $J^{\mu }(\lambda k,q)$, the following differential equations \eqref{diffK}, \eqref{diffJ} need to be solved
\begin{gather}
\frac{{\rm d}K^{\mu }(\lambda k)}{{\rm d}\lambda } =\phi_\alpha^{\mu }(K(\lambda
k))k^{\alpha }=\left(\delta _{\alpha }^{\mu }-\frac{(1-u)}{\kappa }v^{\mu
}K_{\alpha }(\lambda k)\right)\left(1+\frac{u}{\kappa }K(\lambda k)\right)k^{\alpha }, \label{firstK} \\
\frac{{\rm d}J^{\mu }(\lambda k,q)}{{\rm d}\lambda } =\phi_\alpha^{\mu}(J(\lambda k,q))k^{\alpha }=\left(\delta _{\alpha }^{\mu }-\frac{(1-u)}{\kappa }v^{\mu }J_{\alpha }(\lambda k,q)\right)\left(1+\frac{u}{\kappa }J(\lambda k,q)\right)k^{\alpha }, \label{secondJ}
\end{gather}
after using the explicit form of $\phi $. Here the shortcut notation $P=v^{\alpha }P_{\alpha }$ has been extended to $K(\lambda k)=v^{\alpha}K_{\alpha }(\lambda k)$ and $J(\lambda k,q)=v^{\alpha }J_{\alpha }(\lambda k,q)$. The solution of the first equation~\eqref{firstK} is
\begin{gather}
K^{\mu }(k)=k^{\mu }\frac{{\rm e}^{\frac{1}{\kappa }v\cdot k}-1}{\frac{1}{\kappa }v\cdot k}\frac{1}{(1-u){\rm e}^{\frac{1}{\kappa }v\cdot k}+u}. \label{solK}
\end{gather}

The inverse function $\big(K^{-1}\big)^{\mu }(k)$ is
\begin{gather*}
\big(K^{-1}\big)^{\mu }(k)=k^{\mu }\frac{1}{\frac{1}{\kappa }v\cdot k}\ln \left(\frac{1+\frac{u}{\kappa }(v\cdot k) }{1-\frac{(1-u)}{\kappa }(v\cdot k) }\right).
\end{gather*}
The function $J^{\mu }(k,q)$ is calculated similarly. The solution of the second equation \eqref{secondJ} is
\begin{gather}
J_{\mu }(k,q)=\frac{K_{\mu }(k) \big( 1+\frac{u}{\kappa }(v\cdot q) \big) +\big( 1-\frac{(1-u) }{\kappa } ( v\cdot K(k)) \big) q_{\mu }}{1+\frac{u(1-u)}{\kappa ^{2}} ( v\cdot K(k) ) (v\cdot q) }
\label{solJ}
\end{gather}
and $K^{\mu }(k)=J_{\mu }(k,0)$.

For $F_{R,u}$ case, we have $\chi ^{\mu }(P)=\frac{{\rm i}u(1-u)}{\kappa ^{2}}v^{\mu }P$, so we can determine $g(k) $ and $Q(k,q)$ by differentiating with respect to $\lambda $, i.e., using \eqref{diffg}, \eqref{diffQ},
\begin{gather*}
\frac{{\rm d}g(\lambda k)}{{\rm d}\lambda } = k\cdot \chi (K(\lambda k))=\frac{{\rm i}u(1-u)}{\kappa ^{2}} ( k\cdot v ) (v\cdot K(\lambda k)), \\
\frac{{\rm d}Q(\lambda k,q)}{{\rm d}\lambda } = k\cdot \chi (J(\lambda k,q))=\frac{{\rm i}u(1-u)}{\kappa ^{2}} ( k\cdot v ) (v\cdot J(\lambda k,q)).
\end{gather*}
To find the solution to the first equation we use \eqref{solK}
\begin{gather*}
v\cdot K(k)=\frac{\kappa \big( {\rm e}^{\frac{1}{\kappa }v\cdot k}-1\big) }{(1-u){\rm e}^{\frac{1}{\kappa }v\cdot k}+u},
\end{gather*}
and for the second equation we take \eqref{solJ}
\begin{gather*}
v\cdot J(k,q)=\frac{ ( v\cdot K(k) ) \big( 1+\frac{u}{\kappa }(v\cdot q) \big) +\big( 1-\frac{(1-u)
}{\kappa }( v\cdot K(k) ) \big) ( v\cdot q) }{1+\frac{u(1-u)}{\kappa ^{2}}( v\cdot K(k)) (v\cdot q) }.
\end{gather*}
So $g(k)$ and $Q(k,q)$ are calculated as
\begin{gather*}
g(k)=Q(k,0) =\frac{{\rm i}u(1-u)}{\kappa ^{2}} ( k\cdot v ) \int_{0}^{1}{\rm d}\lambda [ v\cdot K(\lambda k) ], \\
Q(k,q) =\frac{{\rm i}u(1-u)}{\kappa ^{2}} ( k\cdot v ) \int_{0}^{1}{\rm d}\lambda [ v\cdot J(\lambda k,q) ].
\end{gather*}
The solutions are
\begin{gather*}
g(k) = Q(k,0)={\rm i}\ln \left( u{\rm e}^{-\frac{(1-u) }{_{\kappa }}v\cdot k}+(1-u) {\rm e}^{\frac{u}{_{\kappa }}v\cdot k}\right), \\
Q(k,q) = {\rm i}\ln \left( u\left( 1-(1-u) \frac{1}{\kappa }v\cdot q\right) {\rm e}^{-\frac{(1-u) }{_{\kappa }}v\cdot k}+ (1-u) \left( 1+u\frac{1}{\kappa }v\cdot q\right) {\rm e}^{\frac{u}{_{\kappa }}v\cdot k}\right).
\end{gather*}
Using equations \eqref{DJK}, it follows that
\begin{gather*}
\mathcal{D}_{\mu }(u;k,q)=\frac{k_{\mu }\big(1+\frac{u}{\kappa }(v\cdot q)\big)+\big(1-\frac{(1-u)}{\kappa }(v\cdot k)\big) q_{\mu }}{1+\frac{u(1-u)}{\kappa ^{2}}(v\cdot k)(v\cdot q)}
\end{gather*}
and
\begin{gather*}
\mathcal{G}(u;k,q)={\rm i}\ln \left( 1+\frac{u(1-u)}{\kappa ^{2}} ( v\cdot k) (v\cdot q) \right),
\end{gather*}
i.e.,
\begin{gather*}
{\rm e}^{{\rm i}\mathcal{G}(u;k,q)}=\frac{1}{1+\frac{u(1-u)}{\kappa ^{2}}( v\cdot k) (v\cdot q) }.
\end{gather*}

Compare the above with \eqref{DFL} and~\eqref{GFR}.

\subsection*{Acknowledgments}

This work has been supported by COST (European Cooperation in Science and
Technology) Action MP1405 QSPACE. AB is supported by Polish National
Science Center (NCN), project UMO-2017/27/B/ST2/01902. We are grateful to Zoran \v{S}koda for his comments.
We would like to thank the referees for their constructive input.

\pdfbookmark[1]{References}{ref}
\LastPageEnding


\begin{thebibliography}{99}
\footnotesize\itemsep=0pt

\bibitem{jhep17}
Aschieri P., Borowiec A., Pacho{\l} A., Observables and dispersion relations in
 {$\kappa$}-{M}inkowski spacetime, \href{https://doi.org/10.1007/jhep10(2017)152}{\textit{J.~High Energy Phys.}} \textbf{2017}
 (2017), no.~10, 152, 27~pages, \href{https://arxiv.org/abs/1703.08726}{arXiv:1703.08726}.

\bibitem{Ballesteros-Jordconf+}
Ballesteros A., Bruno N.R., Herranz F.J., A non-commutative {M}inkowskian
 spacetime from a quantum {A}d{S} algebra, \href{https://doi.org/10.1016/j.physletb.2003.09.014}{\textit{Phys. Lett.~B}} \textbf{574}
 (2003), 276--282, \href{https://arxiv.org/abs/hep-th/0306089}{arXiv:hep-th/0306089}.

\bibitem{Ballesteros-Jordconf++}
Ballesteros A., Bruno N.R., Herranz F.J., Quantum (anti)de {S}itter algebras
 and generalizations of the $\kappa$-{M}inkowski space, in Symmetry Methods in
 Physics, Editors C.~Burdik, O.~Navratil, S.~Po\v{s}ta, Joint Institute for
 Nuclear Research, Dubna, 2008, 1--20, \href{https://arxiv.org/abs/hep-th/040929}{arXiv:hep-th/040929}.

\bibitem{arXiv:0912.3299}
Borowiec A., Gupta K.S., Meljanac S., Pacho{\l} A., Constraints on the quantum
 gravity scale from $\kappa$-Minkowski spacetime, \href{https://doi.org/10.1209/0295-5075/92/20006}{\textit{Europhys. Lett.}}
 \textbf{92} (2010), 20006, 6~pages, \href{https://arxiv.org/abs/0912.3299}{arXiv:0912.3299}.

\bibitem{ijgmmp2016}
Borowiec A., Juri\'c T., Meljanac S., Pacho{\l} A., Central tetrads and
 quantum spacetimes, \href{https://doi.org/10.1142/S0219887816400053}{\textit{Int.~J. Geom. Methods Mod. Phys.}} \textbf{13}
 (2016), 1640005, 18~pages, \href{https://arxiv.org/abs/1602.01292}{arXiv:1602.01292}.

\bibitem{1}
Borowiec A., Lukierski J., Tolstoy V.N., Basic twist quantization of
 $\mathfrak{osp}(1|2)$ and $\kappa$-deformation of $D = 1$ superconformal
 mechanics, \href{https://doi.org/10.1142/S021773230301096X}{\textit{Modern Phys. Lett.~A}} \textbf{18} (2003), 1157--1169,
 \href{https://arxiv.org/abs/hep-th/0301033}{arXiv:hep-th/0301033}.

\bibitem{blt_jhep2017}
Borowiec A., Lukierski J., Tolstoy V.N., Basic quantizations of $D=4$
 Euclidean, Lorentz, Kleinian and quaternionic $\mathfrak{o}^\star(4)$
 symmetries, \href{https://doi.org/10.1007/JHEP11(2017)187}{\textit{J.~High Energy Phys.}} \textbf{2017} (2017), no.~11, 187,
 35~pages, \href{https://arxiv.org/abs/1708.09848}{arXiv:1708.09848}.

\bibitem{BP}
Borowiec A., Pacho{\l} A., $\kappa$-Minkowski spacetime as the result of
 {J}ordanian twist deformation, \href{https://doi.org/10.1103/PhysRevD.79.045012}{\textit{Phys. Rev.~D}} \textbf{79} (2009),
 045012, 11~pages, \href{https://arxiv.org/abs/0812.0576}{arXiv:0812.0576}.

\bibitem{DSRsmash}
Borowiec A., Pacho{\l} A., {$\kappa$}-{M}inkowski spacetimes and {DSR}
 algebras: fresh look and old problems, \href{https://doi.org/10.3842/SIGMA.2010.086}{\textit{SIGMA}} \textbf{6} (2010), 086,
 31~pages, \href{https://arxiv.org/abs/1005.4429}{arXiv:1005.4429}.

\bibitem{bp_jpa2016}
Borowiec A., Pacho{\l} A., Twisted bialgebroids versus bialgebroids from a
 {D}rinfeld twist, \href{https://doi.org/10.1088/1751-8121/50/5/055205}{\textit{J.~Phys.~A: Math. Theor.}} \textbf{50} (2017),
 055205, 17~pages, \href{https://arxiv.org/abs/1603.09280}{arXiv:1603.09280}.

\bibitem{chiari}
Chari V., Pressley A., A guide to quantum groups, Cambridge University Press,
 Cambridge, 1994.

\bibitem{Comtet}
Comtet L., Advanced combinatorics. The art of finite and infinite expansions,
\href{https://doi.org/10.1007/978-94-010-2196-8}{D.~Reidel Publishing Co.}, Dordrecht, 1974.

\bibitem{DJP_SIGMA2012}
Dimitrijevi\'c M., Jonke L., Pacho{\l} A., Gauge theory on twisted
 {$\kappa$}-{M}inkowski: old problems and possible solutions, \href{https://doi.org/10.3842/SIGMA.2014.063}{\textit{SIGMA}}
 \textbf{10} (2014), 063, 22~pages, \href{https://arxiv.org/abs/1403.1857}{arXiv:1403.1857}.

\bibitem{drinfeld}
Drinfeld V.G., Hopf algebras and the quantum {Y}ang--{B}axter equation,
 \textit{Soviet Math. Dokl.} \textbf{32} (1985), 254--258.

\bibitem{etingof}
Etingof P., Kazhdan D., Quantization of {L}ie bialgebras.~{I}, \href{https://doi.org/10.1007/BF01587938}{\textit{Selecta
 Math. (N.S.)}} \textbf{2} (1996), 1--41, \href{https://arxiv.org/abs/q-alg/9506005}{arXiv:q-alg/9506005}.

\bibitem{etingof+}
Etingof P., Schiffmann O., Lectures on quantum groups, 2nd ed., \textit{Lectures in
 Mathematical Physics}, International Press, Somerville, MA, 2002.

\bibitem{Fiore}
Fiore G., On second quantization on noncommutative spaces with twisted
 symmetries, \href{https://doi.org/10.1088/1751-8113/43/15/155401}{\textit{J.~Phys.~A: Math. Theor.}} \textbf{43} (2010), 155401,
 39~pages, \href{https://arxiv.org/abs/0811.0773}{arXiv:0811.0773}.

\bibitem{govindarajan}
Govindarajan T.R., Gupta K.S., Harikumar E., Meljanac S., Meljanac D., Twisted
 statistics in {$\kappa$}-{M}inkowski spacetime, \href{https://doi.org/10.1103/PhysRevD.77.105010}{\textit{Phys. Rev.~D}}
 \textbf{77} (2008), 105010, 6~pages, \href{https://arxiv.org/abs/0802.1576}{arXiv:0802.1576}.

\bibitem{Ballesteros-Jordconf}
Herranz F.J., New quantum (anti)de {S}itter algebras and discrete symmetries,
 \href{https://doi.org/10.1016/S0370-2693(02)02452-8}{\textit{Phys. Lett.~B}} \textbf{543} (2001), 89--97, \href{https://arxiv.org/abs/hep-ph/0205190}{arXiv:hep-ph/0205190}.

\bibitem{tongeren+}
Hoare B., van Tongeren S.J., On jordanian deformations of {${\rm AdS}_5$} and
 supergravity, \href{https://doi.org/10.1088/1751-8113/49/43/434006}{\textit{J.~Phys.~A: Math. Theor.}} \textbf{49} (2016), 434006,
 22~pages, \href{https://arxiv.org/abs/1605.03554}{arXiv:1605.03554}.

\bibitem{kov2}
Juri\'c T., Kova\v{c}evi\'c D., Meljanac S., {$\kappa$}-deformed phase space,
 {H}opf algebroid and twisting, \href{https://doi.org/10.3842/SIGMA.2014.106}{\textit{SIGMA}} \textbf{10} (2014), 106,
 18~pages, \href{https://arxiv.org/abs/1402.0397}{arXiv:1402.0397}.

\bibitem{EPJC2015}
Juri\'c T., Meljanac S., Pikuti\'c D., Realizations of $\kappa$-{M}inkowski
 space, {D}rinfeld twists and related symmetry algebras, \href{https://doi.org/10.1140/epjc/s10052-015-3760-7}{\textit{Eur.
 Phys.~J.~C Part. Fields}} \textbf{75} (2015), 528, 16~pages,
 \href{https://arxiv.org/abs/1506.04955}{arXiv:1506.04955}.

\bibitem{rina-PLA2013}
Juri\'c T., Meljanac S., \v{S}trajn R., {$\kappa$}-{P}oincar\'e--{H}opf algebra
 and {H}opf algebroid structure of phase space from twist, \href{https://doi.org/10.1016/j.physleta.2013.07.021}{\textit{Phys.
 Lett.~A}} \textbf{377} (2013), 2472--2476, \href{https://arxiv.org/abs/1303.0994}{arXiv:1303.0994}.

\bibitem{IJMPA2014}
Juri\'c T., Meljanac S., \v{S}trajn R., Twists, realizations and {H}opf
 algebroid structure of {$\kappa$}-deformed phase space, \href{https://doi.org/10.1142/S0217751X14500225}{\textit{Internat.~J.
 Modern Phys.~A}} \textbf{29} (2014), 1450022, 32~pages, \href{https://arxiv.org/abs/1305.3088}{arXiv:1305.3088}.

\bibitem{MatYosh}
Kawaguchi I., Matsumoto T., Yoshida K., Jordanian deformations of the ${\rm
 AdS}_5 \times {\rm S}_5$ superstring, \href{https://doi.org/10.1007/JHEP04(2014)153}{\textit{J.~High Energy Phys.}}
 \textbf{2014} (2014), no.~4, 153, 20~pages, \href{https://arxiv.org/abs/1401.4855}{arXiv:1401.4855}.

\bibitem{MatYosh++}
Kawaguchi I., Matsumoto T., Yoshida K., A {J}ordanian deformation of {A}d{S}
 space in type {IIB} supergravity, \href{https://doi.org/10.1007/JHEP06(2014)146}{\textit{J.~High Energy Phys.}} \textbf{2014}
 (2014), no.~6, 146, 26~pages, \href{https://arxiv.org/abs/1402.6147}{arXiv:1402.6147}.

\bibitem{kovmel}
Kova\v{c}evi\'c D., Meljanac S., Kappa-{M}inkowski spacetime,
 kappa-{P}oincar\'{e} {H}opf algebra and realizations, \href{https://doi.org/10.1088/1751-8113/45/13/135208}{\textit{J.~Phys.~A:
 Math. Theor.}} \textbf{45} (2012), 135208, 24~pages, \href{https://arxiv.org/abs/1110.0944}{arXiv:1110.0944}.

\bibitem{kmps}
Kova\v{c}evi\'c D., Meljanac S., Pacho{\l} A., \v{S}trajn R., Generalized
 {P}oincar\'{e} algebras, {H}opf algebras and {$\kappa$}-{M}inkowski
 spacetime, \href{https://doi.org/10.1016/j.physletb.2012.03.062}{\textit{Phys. Lett.~B}} \textbf{711} (2012), 122--127,
 \href{https://arxiv.org/abs/1202.3305}{arXiv:1202.3305}.

\bibitem{hreal}
Kova\v{c}evi\'c D., Meljanac S., Samsarov A., \v{S}koda Z., Hermitian
 realizations of $\kappa$-{M}inkowski space-time, \href{https://doi.org/10.1142/S0217751X15500190}{\textit{Int.~J. Geom.
 Methods Mod. Phys.}} \textbf{30} (2015), 1550019, 26~pages,
 \href{https://arxiv.org/abs/1307.5772}{arXiv:1307.5772}.

\bibitem{LukRuegg}
Lukierski J., Nowicki A., Ruegg H., New quantum {P}oincar\'{e} algebra and
 {$\kappa$}-deformed field theory, \href{https://doi.org/10.1016/0370-2693(92)90894-A}{\textit{Phys. Lett.~B}} \textbf{293} (1992),
 344--352.

\bibitem{LukTol}
Lukierski J., Ruegg H., Nowicki A., Tolstoy V.N., {$q$}-deformation of
 {P}oincar\'{e} algebra, \href{https://doi.org/10.1016/0370-2693(91)90358-W}{\textit{Phys. Lett.~B}} \textbf{264} (1991), 331--338.

\bibitem{majid}
Majid S., Foundations of quantum group theory, \href{https://doi.org/10.1017/CBO9780511613104}{Cambridge University Press},
 Cambridge, 1995.

\bibitem{CR}
Mansour T., Schork M., Commutation relations, normal ordering, and {S}tirling
 numbers, \textit{Discrete Mathematics and its Applications (Boca Raton)}, CRC Press,
 Boca Raton, FL, 2016.

\bibitem{MatYosh+}
Matsumoto T., Yoshida K., Lunin--{M}aldacena backgrounds from the classical
 {Y}ang--{B}axter equation~-- towards the gravity/{CYBE} correspondence,
 \href{https://doi.org/10.1007/JHEP06(2014)135}{\textit{J.~High Energy Phys.}} \textbf{2014} (2014), no.~6, 135, 16~pages,
 \href{https://arxiv.org/abs/1404.1838}{arXiv:1404.1838}.

\bibitem{MatYosh+++}
Matsumoto T., Yoshida K., Yang--{B}axter deformations and string dualities,
 \href{https://doi.org/10.1007/JHEP03(2015)137}{\textit{J.~High Energy Phys.}} \textbf{2015} (2015), no.~3, 137, 21~pages,
 \href{https://arxiv.org/abs/1412.3658}{arXiv:1412.3658}.

\bibitem{1903_zagreb}
Meljanac D., Meljanac S., Mignemi S., \v{S}trajn R., $\kappa$-deformed phase
 spaces, {J}ordanian twists, {L}orentz--{W}eyl algebra and dispersion
 relations, \href{https://doi.org/10.1103/PhysRevD.99.126012}{\textit{Phys. Rev.~D}} \textbf{99} (2019), 126012, 12~pages,
 \href{https://arxiv.org/abs/1903.08679}{arXiv:1903.08679}.

\bibitem{pikuticEPJC2017}
Meljanac D., Meljanac S., Pikuti\'c D., Families of vector-like deformations of
 relativistic quantum phase spaces, twists and symmetries, \href{https://doi.org/10.1140/epjc/s10052-017-5373-9}{\textit{Eur.
 Phys.~J.~C Part. Fields}} \textbf{77} (2017), 830, 12~pages,
 \href{https://arxiv.org/abs/1709.04745}{arXiv:1709.04745}.

\bibitem{mkj}
Meljanac S., Kre\v{s}i\'c-Juri\'c S., Differential structure on
 {$\kappa$}-{M}inkowski space, and {$\kappa$}-{P}oincar\'{e} algebra,
 \href{https://doi.org/10.1142/S0217751X11053948}{\textit{Internat.~J. Modern Phys.~A}} \textbf{26} (2011), 3385--3402,
 \href{https://arxiv.org/abs/1004.4647}{arXiv:1004.4647}.

\bibitem{stojic1}
Meljanac S., Kre\v{s}i\'c-Juri\'c S., Stoji\'c M., Covariant realizations of
 kappa-deformed space, \href{https://doi.org/10.1140/epjc/s10052-007-0285-8}{\textit{Eur. Phys.~J.~C Part. Fields}} \textbf{51}
 (2007), 229--240, \href{https://arxiv.org/abs/hep-th/0702215}{arXiv:hep-th/0702215}.

\bibitem{mercati}
Meljanac S., Meljanac D., Mercati F., Pikuti\'c D., Noncommutative spaces and
 {P}oincar\'e symmetry, \href{https://doi.org/10.1016/j.physletb.2017.01.006}{\textit{Phys. Lett.~B}} \textbf{766} (2017), 181--185,
 \href{https://arxiv.org/abs/1610.06716}{arXiv:1610.06716}.

\bibitem{MMPP}
Meljanac S., Meljanac D., Pacho{\l} A., Pikuti\'c D., Remarks on simple
 interpolation between {J}ordanian twists, \href{https://doi.org/10.1088/1751-8121/aa72d7}{\textit{J.~Phys.~A: Math. Theor.}}
 \textbf{50} (2017), 265201, 11~pages, \href{https://arxiv.org/abs/1612.07984}{arXiv:1612.07984}.

\bibitem{mmss1}
Meljanac S., Meljanac D., Samsarov A., Stoji\'c M., Lie algebraic deformations
 of {M}inkowski space with {P}oincar\'e algebra, \href{https://arxiv.org/abs/0909.1706}{arXiv:0909.1706}.

\bibitem{mmss2}
Meljanac S., Meljanac D., Samsarov A., Stoji\'c M., {$\kappa$}-deformed
 {S}nyder spacetime, \href{https://doi.org/10.1142/S0217732310032652}{\textit{Modern Phys. Lett.~A}} \textbf{25} (2010),
 579--590, \href{https://arxiv.org/abs/0912.5087}{arXiv:0912.5087}.

\bibitem{mmss3}
Meljanac S., Meljanac D., Samsarov A., Stoji\'c M., Kappa Snyder deformations
 of {M}inkowski spacetime, realizations and {H}opf algebra, \href{https://doi.org/10.1103/PhysRevD.83.065009}{\textit{Phys.
 Rev.~D}} \textbf{83} (2011), 065009, 16~pages, \href{https://arxiv.org/abs/1102.1655}{arXiv:1102.1655}.

\bibitem{PRD-conformal}
Meljanac S., Pacho{\l} A., Pikuti\'c D., Twisted conformal algebra related to
 {$\kappa$}-{M}inkowski space, \href{https://doi.org/10.1103/PhysRevD.92.105015}{\textit{Phys. Rev.~D}} \textbf{92} (2015),
 105015, 8~pages, \href{https://arxiv.org/abs/1509.02115}{arXiv:1509.02115}.

\bibitem{mssg}
Meljanac S., Samsarov A., Stoji\'c M., Gupta K.S., {$\kappa$}-{M}inkowski
 spacetime and the star product realizations, \href{https://doi.org/10.1140/epjc/s10052-007-0450-0}{\textit{Eur. Phys.~J.~C Part.
 Fields}} \textbf{53} (2008), 295--309, \href{https://arxiv.org/abs/0705.2471}{arXiv:0705.2471}.

\bibitem{stojic2}
Meljanac S., Stoji\'c M., New realizations of {L}ie algebra kappa-deformed
 {E}uclidean space, \href{https://doi.org/10.1140/epjc/s2006-02584-8}{\textit{Eur. Phys.~J.~C Part. Fields}} \textbf{47} (2006),
 531--539, \href{https://arxiv.org/abs/hep-th/0605133}{arXiv:hep-th/0605133}.

\bibitem{svrtan}
Meljanac S., \v{S}koda Z., Svrtan D., Exponential formulas and {L}ie algebra
 type star products, \href{https://doi.org/10.3842/SIGMA.2012.013}{\textit{SIGMA}} \textbf{8} (2012), 013, 15~pages,
 \href{https://arxiv.org/abs/1006.0478}{arXiv:1006.0478}.

\bibitem{Og}
Ogievetsky O., Hopf structures on the {B}orel subalgebra of
 {$\mathfrak{sl}(2)$}, \textit{Rend. Circ. Mat. Palermo~(2) Suppl.} (1994),
 185--199.

\bibitem{pv}
Pacho{\l} A., Vitale P., {$\kappa$}-{M}inkowski star product in any dimension
 from symplectic realization, \href{https://doi.org/10.1088/1751-8113/48/44/445202}{\textit{J.~Phys.~A: Math. Theor.}} \textbf{48}
 (2015), 445202, 16~pages, \href{https://arxiv.org/abs/1507.03523}{arXiv:1507.03523}.

\bibitem{2}
Tolstoy V.N., Chains of extended Jordanian twists for Lie superalgebras,
 \href{https://arxiv.org/abs/math.QA/0402433}{arXiv:math.QA/0402433}.

\bibitem{tolstoy2}
Tolstoy V.N., Quantum deformations of relativistic symmetries,
 \href{https://arxiv.org/abs/0704.0081}{arXiv:0704.0081}.

\bibitem{3}
Tolstoy V.N., Multiparameter quantum deformations of Jordanian type for Lie
 superalgebras, in Differential Geometry and Physics, \textit{Nankai Tracts in
 Mathematics}, Vol.~10, \href{https://doi.org/10.1142/9789812772527_0041}{World Sci. Publ.}, Hackensack, NJ, 2008, 443--452,
 \href{https://arxiv.org/abs/math.QA/0701079}{arXiv:math.QA/0701079}.

\bibitem{tolstoy1}
Tolstoy V.N., Twisted quantum deformations of {L}orentz and {P}oincar\'e
 algebras, \textit{Bulg.~J. Phys.} (2008), suppl.~1, 441--459,
 \href{https://arxiv.org/abs/0712.3962}{arXiv:0712.3962}.

\bibitem{tongeren}
van Tongeren S.J., Yang--{B}axter deformations, {A}d{S}/{CFT}, and
 twist-noncommutative gauge theory, \href{https://doi.org/10.1016/j.nuclphysb.2016.01.012}{\textit{Nuclear Phys.~B}} \textbf{904}
 (2016), 148--175, \href{https://arxiv.org/abs/#2}{arXiv:1506.01023}.

\end{thebibliography}
\end{document}